\documentclass[smallextended,nospthms]{svjour3}

\usepackage{amsthm}  

\usepackage{amsmath}
\usepackage{latexsym}
\usepackage{amssymb}
\usepackage{verbatim}
\usepackage{setspace}
\usepackage{qtree}
\usepackage{pbox}
\usepackage{enumerate}

\usepackage{stmaryrd}    

\usepackage{graphicx}
\usepackage{tikz}
\usetikzlibrary{tikzmark}

\usepackage{tikz}
\usetikzlibrary{arrows,decorations.pathmorphing,backgrounds,positioning,fit}

{\usepackage{txfonts}  
   \DeclareSymbolFont{symbolsC}{U}{txsyc}{m}{n}
   \DeclareMathSymbol{\strictif}{\mathrel}{symbolsC}{74}
   \DeclareMathSymbol{\boxright}{\mathrel}{symbolsC}{128}
}

\newtheorem{df}{Definition}[section]  
\newtheorem{teo}[df]{Theorem}

\newtheorem{lem}[df]{Lemma}

\newtheorem{prop}[df]{Proposition}
\newtheorem{example}[df]{Example}

\newcommand{\dep}[2]{=\hspace{-3pt}({#1};{#2})}

\newcommand{\dfn}{\stackrel{\text{def}}{=}}

\newcommand{\cf}{\boxright}

\newcommand{\PCD}{\mathcal{PCO}}

\newcommand{\PO}{\mathcal{PO}}

\newcommand{\CO}{\mathcal{CO}}

\newcommand{\PCO}{\mathcal{PCO}}

\newcommand{\A}{\mathbb{A}}

\newcommand{\B}{\mathbb{B}}

\newcommand{\F}{\mathcal{F}}
\newcommand{\G}{\mathcal{G}}
\newcommand{\K}{\mathcal{K}}

\newcommand{\La}{\mathcal{L}}

\newcommand{\SET}[1]{\mathbf{#1}}

\newcommand{\pindep}{\rotatebox[origin=c]{90}{$\models$}}

\newcommand{\commf}[1]{}

\newcommand{\corrf}[1]{#1}

\journalname{Journal of Philosophical Logic} 

\author{Fausto Barbero and Gabriel Sandu}

\titlerunning{Multiteam semantics for interventionist counterfactuals}

\title{Multiteam semantics for interventionist counterfactuals: probabilities and causation}

\begin{document}

\setcounter{page}{1} 

\maketitle


\begin{abstract}
In \cite{BarSan2020}, we introduced an extension of team semantics (causal teams) which assigns an interpretation to interventionist counterfactuals and causal notions based on them (as e.g. in Pearl's and Woodward's manipulationist approaches to causation). We now present a further extension of this framework (causal multiteams) which allows us to talk about probabilistic causal statements. We analyze the  expressivity resources of two causal-probabilistic languages, one finitary and one infinitary. 
We show that many causal-probabilistic notions from the field of causal inference can be expressed already in the finitary language, and we prove a normal form theorem that throws new light on Pearl's ``ladder of causation''. On the other hand, we provide an exact semantic characterization of the infinitary language, which shows that this language captures precisely those causal-probabilistic statements that do not commit us to any specific interpretation of probability; and we prove that no usual, countable language is apt for this task.\\

\end{abstract}

\section{Introduction}



In their new book, \emph{The Book of Why: The New Science of Cause and Effect} (\cite{PeaMac2018}), Pearl and MacKenzie popularized the idea of structuring the tasks of causal inference into three levels, under the name of ``ladder of causation''.\footnote{In the technical literature, this idea was introduced (as \emph{causal hierarchy}) already in \cite{ShpPea2008}.} 

The first level, which is the most elementary, concerns associations among phenomena or, from an operational perspective, dependencies in the data.\footnote{In popularization books, the word \emph{correlation} is often used for what we call here association. Since the word correlation is also used in the statistical sciences with a more specific meaning, we refrain from using it in this paper.} For technical convenience, phenomena are described in terms of variables (e.g. pressure, temperature...) taking values. Also \emph{events}, such as ``suffering of a certain disease'' or ``being treated'', can be described in this way by using Boolean variables, say a variable  $E$ that takes value $1$ if the event occurs, and value $0$ if the event does not occur. A typical query arising at this stage is e.g. ``Given that a patient suffers from a particular
disease ($Sick = 1$), what is the probability that he dies ($Dies = 1$), if he is not receiving treatment ($Treated = 0$)?'' Such queries are usually answered just by collecting data and analyzing it by means of the rules of conditional probability. 

The second level (interventions) involves 
 an estimation of the consequences of performing a given action. A query at this second level could be: “What is the probability that a patient recovers, if we start treating him?” Thus the query forces us to calculate the impact on the patient of the action of treating him. An answer to such a question (e.g., that such a probability is $\epsilon$) is stated in the literature as a \emph{do expression}:
\[
\Pr(Dies =0 \mid do(Treated=1)) = \epsilon.
\]
Pearl also includes in the second level what may be called \emph{conditional do} expressions, which have a more general form: $\Pr(X=x \mid do(Y=y), Z=z)$. The intended referent of this expression is the probability of the event $X=x$ after an action that forces $Y$ to take value $y$, conditional on the event ``$Z$ takes value $z$ after the action''. The conditional \emph{do} expressions have been mainly used as a mathematical tool for computing the probabilities of \emph{un}conditional \emph{do} expressions in the context of Pearl's \emph{do} calculus. Nonetheless, there are interesting queries that may be formulated about them; for example, ``If a patient were treated, what's the probability that he will abandon treatment if he develops a certain side effect?'' An answer is expressed in symbols as
\[
\Pr(Abandons = 1 \mid do(Treated=1), SideE = 1) = \epsilon.
\]
 
  The second level is obviously more complex, for the impact an action might have depends on the causal connections among  
   the variables in the system. Interventions are also essential for framing the next level.


The third level (imagining) involves reasoning about hypothetical situations, things that could happen were certain actions to be performed (in the sense of the second level), and which may possibly be in direct contradiction with our knowledge base of actual facts. Pearl calls ``counterfactuals'' those statements that answer queries of this form, and we will refer to them as \emph{Pearl counterfactuals}. A typical example (from \cite{Pea2000}) of a Pearl counterfactual is
\begin{center}
(PC): The probability that a subject who died after treatment would have recovered
had he or she not been treated is $\epsilon$.
\end{center}
 The typical formal notation for (PC) is:
\[
\Pr(Dies_{Treated=0}=0 \mid Dies=1) = \epsilon.
\]

Our ambition, which underlies this paper, is to build a framework that expresses the main notions involved in the three levels of Pearl's ladder of causation and the formal relationships between them. We started this work in  \cite{BarSan2020}, where we constructed a logical framework (\emph{causal team semantics}) which provides a model-theoretic treatment of observations, interventions, and interventionist counterfactuals,  
and we showed how to define various notions of determinitic cause in the Woodward-Hitchcock
style (\cite{Hit2001},\cite{Woo2003}). There our main interest was to distinguish between accidental associations and causal, robust associations which support counterfactual reasoning, possibly
conditional on observations. To this effect we combined two existing approaches,
\emph{team semantics} and \emph{structural equations models}, and formulated the notion
of a \emph{causal team}. Essentially, a causal team is a team (a set
of assignments over a given collection of variables) expanded with a collection
of structural equations associated with some of its variables. Due to its team
component, a causal team is able to account for associations 
and for evidential, observational reasoning; and due to its structural
equations component, a causal team is able to account for interventions
and support counterfactual claims. Thus, with the notion of a causal team we
have all the ingredients needed for the analysis of the kind of mixed statements that are 
involved in Pearl's ladder of causation, except probabilities.

The main purpose of the present paper is to extend the range of
applications of causal team semantics, and its associated notion of interventionist
counterfactual, to languages that allow also the discussion of probabilities.
From a technical point of view, this requires shifting attention from teams (sets
of assignments) to multiteams (multisets of assignments). For the purpose of illustrating the characteristics of our approach, we choose first to focus on a relatively poor language called $\PCO$, 
 which has among its primitive expressions some basic probabilistic atoms, but no arithmetical operations. The analysis of this language already brings to light many important facts. For example, the multiteam framework quickly makes clear that two distinct species of logical connectives 
 are involved when discussing events or, instead, when discussing the events' probabilities.\footnote{The roles of negation, disjunction and implication will be taken by connectives $\neg, \lor,\supset$ at the level of events, and by distinct connectives $^C, \sqcup,\rightarrow$ at the level of probabilities. On the other hand, the connectives $\land$ and $\cf$ (for conjunction and counterfactual) are adequate at both levels.} 
 Secondly, the interaction of probabilistic atoms, selective implication ($\supset$) and interventionist
counterfactuals ($\cf$) allows us to express mixed causal-probabilistic
statements such as (PC) and other types of statements coming from the literature
on causal modeling or going beyond its usual scope. 
 In particular, we show that our language expresses in a perspicuous way the distinction between probabilities which condition on observations of a post-intervention system and probabilities which condition on observations of the pre-intervention system. This is crucial for distinguishing the meaning of $do$ expressions from that of Pearl counterfactuals, but it also paves the way for the study of more complex combinations of causal and probabilistic reasoning. Thirdly, our framework will allow us to take a stand on a recent debate between Pearl and Tim Maudlin. In his review (\cite{Mau2019}) of Pearl's and MacKenzie's book, Maudlin criticized Pearl's ``ladder of causation'' on the ground that the distinction between levels two and three does not seem to be justified. In his answer to Maudlin’s criticism, Pearl pointed out that although the distinction between level two and level three may not be obvious in the book, these levels differ strongly in a specific technical sense. We will review the debate (section \ref{sec: debate}) and show how our result support some of Maudlin's criticisms.

We will also have a look at what lies beyond $\PCO$; in particular, we inquire how a language could be obtained that allows to express statements about probabilistic interventionist counterfactuals in full generality. We show that such a language must be uncountable, and that the goal can be achieved simply by adding to $\PCO$ a (countably) infinitary disjunction operator. 
 More precisely, this extended language $\PCO^\omega$ characterizes the set of those statements that do not require us to take a stance about our choice of interpretation of the concept of probability; on the other hand, we show that an unwise choice of atoms and logical operators can instead force upon us a frequentist interpretation.


\textbf{Structure of the paper.} In section \ref{sec: preliminaries} we decribe the models of causal multiteam semantics, which we then use in section \ref{sec: level of events} to formalize a language $\CO$ for the (qualitative) discussion of events. In section \ref{sec: level of probabilities} we then use $\CO$ as a basis over which we build a probabilistic language $\PCO$, and we analyze in some detail what additional operators are definable in this language. In section \ref{sec: comparison with causal models} we briefly explain the connections between causal multiteam semantics and the more conventional  causal models. In section \ref{sec: definable concepts} we show how the two conditionals $\supset$ and $\cf$ can be used to build formulas that express many different kinds of probabilistic-causal concepts; on one hand, we use them to clarify the notations used in the literature on causal inference, and on the other hand we illustrate how they can be used to express even more complex causal concepts. Running against this proliferation, a normal form theorem in section \ref{subs: normal form} shows that this variety of notions can be reduced to (Boolean combinations of) three simple kinds of formulae, corresponding to probabilistic conditioning, \emph{un}conditional $do$ expressions and Pearl counterfactuals. In section
\ref{sec: infinitary language} we observe that a simple infinitary extension of $\PCO$, called $\PCO^\omega$, is expressively complete if we assume that a multiteam merely encodes a probability distribution; while extending $\PCO$ with other, seemingly innocent operators enforces a frequentist interpretation of the multiteams as collections of data. In section \ref{sec: debate} we use the apparatus developed so far and take a stand on the exchange between
Pearl and Maudlin concerning the 3 levels in the “ladder of causation”. The Appendix 
 contains the details of the proof of the semantic characterization theorem for $\PCO^\omega$.

\section{Semantic preliminaries} \label{sec: preliminaries}

\subsection{Teams and multiteams}

The protagonists of our paper will be \emph{variables}, which must not be thought of, as
they most commonly are in the logical literature, as mere, perfectly interchangeable
placeholders. Each variable (denoted by a capital letter, such as $X, Y, Z, U, V,\dots$, or by a capitalized word)
should instead be thought of, on the model of empirical sciences, as standing
for a specific magnitude (e.g. ``temperature'', ``volume'', ``position'') which can
\emph{take values} (denoted in small letters, e.g. the values of the variable $X$ will be
denoted by $x, x', x''\dots$) As a special case, we may also consider Boolean variables
which take values $1$ or $0$ to describe the occurrence, resp. non-occurrence,
of a specific event. We make an exception to these notational conventions and
reserve the capital letters $S$ and $T$ to denote causal multiteams, to be defined below. Boldface letters such as $\SET X$ or $\SET x$ will be used to denote either tuples or sets of variables (resp. values).

By a \textbf{signature} we will mean a pair $(Dom,Ran)$, where $Dom$ is a finite set of variables and $Ran$ is a function that associates to each variable $X\in Dom$ a finite set $Ran(X)$ of values (the \emph{range} of $X$). 
An \textbf{assignment} of signature $\sigma$ will be a mapping $s:Dom\rightarrow\bigcup_{X\in Dom}Ran(X)$ such that $s(X)\in Ran(X)$
for each $X\in Dom$. We will denote as $\B_\sigma$ the set of assignments of signature $\sigma$. A \textbf{team} $T$ of signature $\sigma$ will be any set of such assignments, i.e. any subset of $\B_\sigma$. Here is an example of a team $T = \{s_1, s_2, s_3\}$ with 
$Dom(T) = \{X, Y,Z,W\}$ and range $\{0, 1\}$ for each variable: 
\begin{center}
\begin{tabular}{|c|c|c|c|}
\hline
 \multicolumn{4}{|c|}{\hspace{-4pt}  X \   \ Y  \ \   Z \ \     W} \\
\hline
 0 & 0 & 0 & 1\\
\hline
 0 & 0 & 1 & 0 \\
\hline
 0 & 1 & 1 & 0 \\ 
\hline
\end{tabular}
\end{center}
Each row in this picture represents an assigment; for example, the first row encodes the assignment $s_1(X)=0$, $s_1(Y)=0$, $s_1(Z)=0$, $s_1(W)=1$. \commf{If we want to shorten the paper, I think we can skip the motivational example based on nonlocality.} In this paper we will be interested in the probability of events (as formulated in a language) in a given team. Such a probability is assessed by counting. Thus, in the team above, the fact that two assignments out of three assign $0$ to the variable $Y$ encodes the fact that the probability of the event $Y = 0$ in the team $T$ is $\frac{2}{3}$. But it is soon evident that this kind of modeling is inadequate. If we choose to ignore the data concerning some of the variables (say, $Z$ and $W$) we obtain the smaller team:
\begin{center}
\begin{tabular}{|c|c|}
\hline
 \multicolumn{2}{|c|}{\hspace{-4pt}  X \   \ Y  } \\
\hline
 0 & 0 \\
\hline
 0 & 1 \\
\hline
\end{tabular}
\end{center}
and if we proceed by counting, as before, we obtain that the probability of the event $Y=0$ is $\frac{1}{2}$. The deletion of a part of the team that should be irrelevant for the evaluation of $Y=0$ has changed the probability of the event; the problem   here is that the different choice of variables has collapsed the three distinct assignments into two. One way to remedy to this problem is to replace the teams with multisets of assignments, known in the literature as \textbf{multiteams}. There at least two different approaches to the formalization of multiteams in the literature (\cite{Vaa2017},\cite{DurHanKonMeiVir2016}). However, the complex notational apparatus
used in these two papers is mainly needed in order to handle quantification. We
will not use quantifiers in this paper; therefore we can follow a third, simpler
approach, which consists of simulating multiteams by means of teams, and thus
preserving the idea of a multiteam still being a set. This can be accomplished
by assuming that each team has an extra variable $Key$ (which is not part of the
signature, and is never mentioned in the object languages) which takes distinct
values on distinct assignments of the same team; 
\corrf{we will always assume (as an exception) that
$Key$ takes values in $\mathbb N$}. 
 In this way, we can have two assignments that agree on all
the significant variables and just differ on $Key$. If $T$ is treated as a multiteam,
then, we shall abuse notation and write $Dom(T)$ for the set of all variables of
$T$ except $Key$. The set of all extended assignment of signature $\sigma$ (whose domain is actually $Dom\cup \{Key\}$) will be denoted as $\A_\sigma$. \corrf{A \emph{finite} subset of $\A_\sigma$ will be called a \textbf{multiteam} of signature $\sigma$.}

For an illustration, here is how we render the team $T$ above as a multiteam in the new
setting:
\begin{center}
$S$: \begin{tabular}{|c|c|c|c|c|}
\hline
 \multicolumn{5}{|c|}{\hspace{-4pt} Key \  X \   \ Y  \ \   Z \ \     W} \\
\hline
 \phantom{a}0\phantom{a} & 0 & 0 & 0 & 1\\
\hline
 1 & 0 & 0 & 1 & 0 \\
\hline
 2 & 0 & 1 & 1 & 0 \\ 
\hline
\end{tabular}
\end{center}

After deleting the variables Z and W and their corresponding columns, we
obtain the multiteam:
\begin{center}
$S'$: \begin{tabular}{|c|c|c|}
\hline
 \multicolumn{3}{|c|}{\hspace{-4pt} Key \ \ X \   \ Y} \\
\hline
 \phantom{a}0\phantom{a} & 0 & 0 \\
\hline
 1 & 0 & 0 \\
\hline
 2 & 0 & 1 \\ 
\hline
\end{tabular}
\end{center}
As we see, the probability of the event $Y = 0$ in $S'$ is the same as in $S$, namely $\frac{2}{3}$. In general, in the multiteam setting the probability of an event will be a rational number in the interval $[0,1]$.\footnote{More general treatments of probabilities in team semantics have been studied, which allow for irrational values and for infinite models (\cite{HytPaoVaa2017}, \cite{DurHanKonMeiVir2018}). For the purpose of our discussion, rational-valued distributions over finite sets will suffice.}

A multiteam encodes the level of probabilities and associations, i.e. what Pearl calls the rung one of the ladder of causation. The second and third rung will be accessible only by using richer models, to which now we turn.

\subsection{Causal multiteams} \label{subs: causal multiteams}

We introduce here the notion of \emph{causal multiteam}, which can be thought of as a multiteam enriched with the kind of causal structure that might be encoded in a system of structural equations.  There will be (at least) two natural interpretations for such an object, corresponding to a subjectivist, resp. a frequentist point of view. By the former, a causal multiteam will encode probabilistic uncertainty among many possible configurations of the variables, under a fixed hypothesis about the correct causal laws. By the latter interpretation, a causal multiteam will be seen as a collection of empirical data, paired with an assumption about the process that generated such data.\footnote{It may be unrealistic, in many circumstances, that the inquirer has trust in a unique, full description of causes/processes. An approach for treating partial descriptions of the causal laws was suggested in \cite{BarSan2018} (\emph{partially defined causal teams}), and a more general and much more detailed proposal is developed in \cite{BarYan2022} (\emph{generalized causal teams}).  In this first detailed analysis of the probabilistic semantics, we choose the simplicity of the present framework.}

We will consider throughout the paper a fixed ordering of $Dom$, and write $\SET W$ for the tuple of all variables of $Dom$ listed in such order. Furthermore, we write $\SET W_X$ for the variables of $Dom\setminus\{X\}$ listed according to the appropriate restriction of the fixed order. 
Given a tuple $\SET X = (X_1,\dots, X_n)$ of variables, we denote as $Ran(\SET X)$ the cartesian product $Ran(X_1)\times\dots\times Ran(X_n)$. Given an assignment $s$ that has the variables of $\SET X$ in its domain, $s(\SET X)$ will denote the tuple $(s(X_1), \dots,s(X_n))$.

\begin{df}
A \textbf{causal multiteam} $T$ of signature $(Dom(T), Ran(T))$ 
with endogenous variables $\mathbf V\subseteq Dom(T)$ is a pair $T = (T^-,\F)$, where:
\begin{enumerate}
\item $T^-$ is a multiteam of domain $Dom(T)$ (the \emph{multiteam component} of $T$)
\item $\F$ is a function $\{(Y,\F_Y) \ | \ Y\in\SET V\}$ (the \emph{function component} of $T$) that assigns to each endogenous variable $Y$ a non-constant $|\SET W_Y|$-ary function $\F_Y: Ran(\SET W_Y)\rightarrow 
Ran(Y)$ 
\end{enumerate}
which satisfies the further \textbf{compatibility constraint}:

\begin{itemize}
\item For all $s\in T^-$, and all $Y\in \SET V$, $s(Y)= \F_Y(s(\SET W_Y))$.
\end{itemize}
\end{df}


The function component $\mathcal F$ induces an associated system of structural equations, say one equation
\[
Y := \mathcal F_Y(\SET W_Y)
\]
for each variable $Y\in \SET V$. We point out that that some of the variables in $\SET W_Y$ may not really be needed for evaluating $Y$. For example, if $Y$ is given by the structural equation $Y:= X+1$, all variables in $\SET W_Y\setminus \{X\}$ are irrelevant; we say that they are \textbf{dummy arguments} of the function $\F_Y$. The set of arguments of $\F_Y$ that are not dummy is denoted as $PA_Y$ (the set of \textbf{parents} of $Y$). 

We can associate to each causal multiteam $T$ a \textbf{causal graph} $G$, whose vertices are the variables in $Dom$ and where an arrow is drawn from each variable in $PA_Y$ to $Y$, whenever $Y$ is endogenous.  
The variables in $Dom(T)\setminus \SET V$ will be called, as usual, \textbf{exogenous}, and they are not the endpoint of any arrow in the graph.\footnote{Our presentation of causal (multi)teams here diverges in a number of ways from what was done in previous papers, such as \cite{BarSan2020}. There the causal graph is part of the definition, and its presence allows for the introduction of functions $\F_V$ with (almost) arbitrary sets of arguments $PA_V$. Such a definition allows the representation of causal laws that only differ for the presence or absence of some dummy arguments. This detail turns out to be a rather tedious technicality (see \cite{BarYan2022} for a detailed analysis) and we choose to avoid it in this paper.}
In this paper, we will typically allow only acyclic graphs. A causal multiteam with acyclic graph will be said to be \textbf{recursive}.\footnote{By analogy with the terminology used in structural equation modeling. The intuition is that, if the graph is (well-founded and) acyclic, the value  of any variable can be reconstructed by recursively applying the causal functions, once the values of the exogenous variables are known. Since we consider only finite graphs, these are automatically well-founded.} 


Let us consider an example of a (recursive) causal multiteam. Consider the table:

\begin{center}
$T^-$: \begin{tabular}{|c|c|c|}
\hline
\multicolumn{3}{|l|}{Key \hspace{3pt} $X$ \ \ $Y$} \\
\hline
 \phantom{a}0\phantom{a} & 0 & 1 \\
\hline
 1 & 1 & 2 \\
\hline
 2 & 1 & 2 \\
\hline
 3 & 2 & 3 \\ 
\hline
 4 & 2 & 3 \\ 
\hline
 5 & 2 & 3 \\ 
\hline
\end{tabular}
\end{center}
If we pair this table with a function component $\F=\{(Y,\F_Y)\}$ that associates to variable $Y$ the function $\F_Y(X):= X+1$, we obtain a causal multiteam $T=(T^-,\F)$. The reason why this is a causal multiteam is that each assignment $s$ in it is compatible with the equation, in the sense that $s(Y)=s(X)+1$. The causal graph associated with this causal multiteam is $X \rightarrow Y$. 
\section{The level of events} \label{sec: level of events}

We briefly review the (non-probabilistic) language $\CO$, which was considered in \cite{BarSan2020}, and its causal multiteam semantics, which can be adapted from teams to multiteams word-per-word. This language actually resembles the non-probabilistic languages  used in the literature on causation.  Later, a more general language will be used to discuss the probabilities of $\CO$ formulas. The language $\CO$ is, in the technical sense by which the word is used in the literature on team semantics, a \emph{flat} language. This means that the truth value of a $\CO$ formula over a multiteam can in principle always be reduced to its truth values over single assignments (enriched with causal structure). Therefore, we could in principle give Tarskian semantical clauses for this language. We prefer, however, to stick to team semantics for compatibility with the rest of the paper, and also because the team perspective makes it clearer what is the intended meaning of some of the operators (in particular, the ``selective implication''). 

\begin{df} \label{def: CO language}
For a fixed signature $\sigma = (Dom,Ran)$, we call  $\CO_{\sigma}$ (i.e. \emph{causal-observational} language) the language formed according to the following rules:
\[
\alpha ::= Y=y \mid Y\neq y \mid  \alpha\land\alpha  \mid  \alpha\lor\alpha  \mid  \alpha\supset\alpha  \mid  \SET X = \SET x \cf \alpha
\]
for $\SET{X}\cup\{Y\}\subseteq Dom$, $y\in Ran(Y),\SET x\in Ran(\SET X)$. The notation $\SET X = \SET x$ abbreviates a formula of the form $X_1 = x_1 \land\dots\land X_n = x_n$.
We will often simply write $\CO$, omitting reference to the specific signature. $\CO$ formulas of the forms $Y = y$ and $Y\neq y$ will be called \emph{literals}.
\end{df}
\noindent We have not included a negation symbol in the syntax of $\CO$ languages, since negation is definable (see the remarks after theorem \ref{thm: empty multiteam}).

We briefly review the semantics of this language. Some semantic clauses are straightforward:

\begin{itemize}
\item $T\models Y=y$ if, for all $s\in T^-$, $s(Y)=y$.
\item $T\models Y\neq y$ if, for all $s\in T^-$, $s(Y)\neq y$.
\item $T\models \alpha\land \beta$ if $T\models \alpha$ and $T\models \beta$.
\end{itemize}

The semantic clause for $\lor$ (\emph{tensor disjunction}) can more easily be defined in terms of the following notion of submodel.
\begin{df}
Given a causal multiteam $T=(T^-,\F_T)$, a \textbf{causal sub-multiteam} $S=(S^-,\F_S)$ of $T$ is a causal multiteam with the same domain and the same set of endogenous variables, which satisfies the following conditions: 1) $S^-\subseteq T^-$, 2) $\mathcal{F}_S = \mathcal{F}_T$.
\end{df} 
\noindent The definition immediately entails that, in order to identify a causal sub-multiteam $S$ of $T$, it is sufficient to specify its multiteam component. We can now define the semantics of $\lor$.

\begin{itemize}
\item $T\models \alpha\lor \beta$ if there are two causal sub-multiteams $T_1,T_2$ of $T$ such that $T_1^-\cup T_2^- = T^-$, $T_1\models \alpha$ and $T_2\models \beta$.
\end{itemize}

\noindent The seemingly unnatural clause for $\lor$ is needed in order  to extend ``conservatively'' the meaning of disjunction from Tarskian to team semantics\footnote{\label{footnote: conservative extension}For notational simplicity, let us think for a moment of (multi)teams without causal structure. What we mean by this conservativity statement is that, for any singleton team $\{s\}$, $\{s\}\models \alpha \lor \beta$ if and only if $s\models \alpha \lor \beta$ in the Tarskian sense.} while preserving flatness. 

The \emph{selective implication} $\alpha\supset\beta$ is used to assert the statement $\beta$  conditional on the observation $\alpha$. Each observation (antecedent) $\alpha$ corresponds to an operator on teams:

\begin{df}
Given a causal multiteam $T = (T^-,\F)$, and a $\CO$ formula $\alpha$, we define the causal sub-multiteam $T^\alpha$ by the condition:
\[
(T^\alpha)^- \dfn \{s\in T^-  \ | \  (\{s\},\F)\models \alpha\}.
\]
\end{df}

\noindent Then the semantic clause for $\supset$ is:

\begin{itemize}
\item $T\models \alpha \supset \beta$ iff $T^\alpha \models \beta$.
\end{itemize}
\corrf{Notice that the definition of $\models$ and the semantic clause for $\supset$ are interdependent, but since there are only finitely many occurrences of $\supset$ in a formula, these definitions are not circular.}

The interventionist counterfactual implication $\SET X = \SET x \cf \alpha$ is more complex to define. Intuitively 
 the idea is that we can associate to each causal multiteam $T = (T^-,\F)$ and each \emph{consistent}\footnote{In our context, a conjunction $\SET X = \SET x$ is said to be inconsistent if it contains two conjuncts of the form $X=x$, $X=x'$ for some variable $X\in \SET X$ and two distinct values $x\neq x'$. Otherwise, it is consistent.} antecedent $\SET X = \SET x$ an ``intervened causal multiteam'' $T_{\SET X = \SET x}$ which, we say, has ``undergone the intervention $do(\SET X = \SET x)$''. In the recursive case, which is the only one we consider here, this intervention induces the following changes on the multiteam component of $T$: 1) the $\SET X$-columns assume constant value $\SET x$, and 2) all the columns corresponding to descendants of $\SET X$ (i.e. variables that can be reached from $\SET X$ by following a directed path on the graph) are updated by applying the functions from $\F$ in an appropriate order. By this we simply mean that each function must be applied only after all its (non-dummy) arguments have been recomputed, if they need to be (see \cite{BarSan2020} for more details). The resulting multiteam is denoted as $T_{\SET X = \SET x}^-$. The intervention $do(\SET X = \SET x)$ also induces the following changes on other components of $T$: 3) all variables in $\SET X$ become exogenous; thus, if any variable $X$ of $\SET X$ was endogenous, now its corresponding function $\F_X$ is removed from  $\mathcal F$. This restriction of the function component $\F$ is denoted as $\F_{\SET X = \SET x}$. 
 
 We remark that the procedure described above also modifies the causal graph $G$ of $T$: all arrows entering any of the variables of $\SET X$ are removed from the graph $G$, obtaining a new graph $G_{\SET X=\SET x}$.

More formally:

\begin{df}
Let $T=(T^-,\F)$ be a recursive causal multiteam of signature $\sigma$ and endogenous variables $\SET V$. Let $\SET X = \SET x$ be a consistent conjunction. The intervention $do(\SET X = \SET x)$ produces another causal multiteam $T_{\SET X = \SET x}=(T_{\SET X = \SET x}^-,\F_{\SET X = \SET x})$ of signature $\sigma$ and endogenous variables $\SET V \setminus \SET X$ as follows:  
\begin{itemize}
\item $\F_{\SET X = \SET x} \dfn \F_{\upharpoonright(\SET V \setminus \SET X)}$ \hspace{5pt} \corrf{(the restriction of $\F$ to $\SET V \setminus \SET X$)}
\item $T_{\SET X=\SET x}^- \dfn \{s^\F_{\SET X=\SET x}\mid s\in T^-\}$, where each $s^\F_{\SET X=\SET x}$ is the unique assignment compatible with $\mathcal F_{{\SET X = \SET x}}$ defined (recursively) as 
\[s^\F_{\SET X=\SET x}(V)=\begin{cases}
x_i&\text{ if }V=X_i\in \SET X\\
s(V)&\text{ if }V\in Exo(T)\setminus \SET X\\
\F_V(s^\F_{\SET X=\SET x}(\SET W_V
))&\text{ if }V\in End(T)\setminus \SET X
\end{cases}\]
\end{itemize}
\end{df}
\noindent It is worth noting that an intervention never changes the cardinality of a causal multiteam (contrarily to what may happen with causal teams).

The interventionist counterfactual is then assigned its meaning by the following clause:

\begin{itemize}
\item $T\models \SET X=\SET x \cf \psi \iff T_{\SET X=\SET x} \models \psi$ or $\SET X=\SET x$ is inconsistent.
\end{itemize}

The $\CO$ language allows us to express counterfactual statements together with the context in which they are typically embedded, e.g. the following non-probabilistic variant of our earlier example (PC): 
\begin{quote}
A subject who died after treatment ($Dies=1,Treated=1$) would have recovered ($Dies=0$) had he or she not been treated ($Treated=0$).
\end{quote}
We render this statement in the $\CO$ language as:
\[
(Dies=1 \land Treated=1)\supset (Treated =0 \cf Dies=0).
\]
Later we will introduce a language  that can express probabilistic variants of the above statement.

Notice also that the $\CO$ language allows counterfactuals to occur as antecedents of the selective implication, as in the formula $(\SET X = \SET x \cf \alpha)\supset \chi$. The intuition here is that we can assert that $\chi$ is true once we know a certain outcome of the ``experiment'' of fixing $\SET X$ to $\SET x$.\footnote{See \cite{BarSchVelXie2021} for a discussion of this perspective, and of the distinction between the logic of ideal and real experiments.}

The key property of language $\CO$ is \emph{flatness}, which in a sense characterizes it as a classical language.

\begin{teo}[Flatness]
Let $\alpha$ be a $\CO_{\sigma}$ formula, and $T = (T^-,\F)$ a causal multiteam of signature $\sigma$. Then $T\models\alpha$ if and only if, for all $s\in T^-$, $(\{s\},\F)\models\alpha$.
\end{teo} 
\noindent The proof is similar to that of the analogous result for causal teams (\cite{BarSan2020}, theorem 2.10). 
It immediately entails a second property. We say that a causal multiteam $T = (T^-,\F)$ is \textbf{empty} (resp. \textbf{nonempty}) if such is the multiteam $T^-$ (which we have identified with a set of assignments over $Dom\cup{Key}$).

\begin{teo}[Empty multiteam property]\label{thm: empty multiteam}
Let $T$ be an empty causal multiteam of signature $\sigma$ and $\alpha$ a $\CO_\sigma$ formula. Then $T\models \alpha$. 
\end{teo}

We conclude this section with some remarks on negation. We did not add negation explicitly to the definition of $\CO$ because it is definable in it. The kind of negation in question is the so-called dual negation:
\begin{itemize}
\item $T\models \neg\alpha \iff $ for all $s\in T^-$, $\{s\}\not \models \alpha$.
\end{itemize}
which is a ``conservative extension'' of the classical negation of Tarskian semantics (see footnote \ref{footnote: conservative extension}).
Let us define $\bot$ as $X=x\cf X\neq x$. Then, it is easy to see (using flatness) that $\neg\alpha$ is equivalent to $\alpha\supset\bot$.\footnote{\label{footnote: dual}Alternatively, \cite{BarSan2020} describes an inductive \emph{dualization} procedure that associates to each $\CO$ formula $\alpha$ a $\CO$ formula $\alpha^d$ that is equivalent to $\neg\alpha$. This alternative definition can be useful for specific applications.} On the other hand, also $\supset$ is definable in terms of $\neg$, since $\alpha\supset \beta$ is equivalent to $\neg\alpha\lor\beta$.\footnote{In more general languages, a selective implication $\alpha\supset\chi$ can be defined as $\neg\alpha \lor (\alpha\land \chi)$.  The simpler form shown in the main text is available due to the flatness of $\CO$.}\footnote{This approach also works with dualization (see footnote \ref{footnote: dual}), i.e. by defining $\alpha\supset \beta$ as $\alpha^d\lor\beta$, since $\alpha^d$ does not contain any occurrence of $\supset$.}


\section{The level of probabilities}\label{sec: level of probabilities}

\subsection{Probabilities over (causal) multiteams}   \label{subs: probabilities}

We want to express, in a logical formalism, probabilistic variants of the statement introduced above. Recall for example the statement (PC):
\begin{quote}
The probability that a subject who died after treatment would have recovered had he or she not been treated is $\epsilon$.
\end{quote} 
We take it to be expressed by a formula such as 
\[
(Dies=1 \land Treated=1)\supset (Treated=0 \cf \Pr(Dies=0)=\epsilon);
\]
thus, we need a language which is capable of expressing statements of the form $\Pr(Y=0)=\epsilon$.

There are multiple ways of introducing probabilities over (multi)teams. In this paper we will follow the most straightforward way, using the counting measure over multiteams (as is done e.g. in \cite{DurHanKonMeiVir2016} and \cite{Vaa2017}). 

For any $\CO$ formula $\alpha$ and any \emph{nonempty} causal multiteam $T=(T^-,\F)$, 
 we define the \textbf{probability of $\alpha$ in $T$} as:
\[
P_T(\alpha)\dfn \frac{|(T^\alpha)^-|}{|T^-|} = \frac{|\{s\in T^-  \ | \  (\{s\},\F)\models \alpha\}|}{|T^-|}
\]
where as usual the notation $|\cdot|$ denotes the cardinality of its argument. We point out that this definition can be extended to cover \emph{any} causal sub-multiteam $S=(S^-,\F)$ of the fixed (nonempty) causal multiteam $T$, not just the definable ones:
\[
P_T(S) \dfn \frac{card(S^-)}{card(T^-)}
\]
and that this latter definition obviously induces a probability space over $T^-$. We will now show that also the space of $\CO$-definable submultiteams induces a probability space. 
These two facts will guarantee the possibility of using the usual rules of probability in the arguments of this paper. 

Let $\mathcal{E}_T$ be the set of all subsets $S$ of $T^-$ which are definable by a $\CO$ formula $\alpha$ (that is, $S= (T^\alpha)^-=\{s\in T^-  \ | \  (\{s\},\F)\models \alpha\}$). 
This will be our set of \emph{events}. By the definitions above, 
 $P_T(S)= P_T(\alpha)$. 



\begin{teo}
Let $\sigma = (Dom,Ran)$ be a signature, with $Dom\neq\emptyset$ and $Ran(X)\neq\emptyset$ for some $X\in Dom$. Let $T$ be a finite causal multiteam of signature $\sigma$. Then $(T^-,\mathcal{E}_T,
 P_T)$ is a probability space, that is:
\begin{enumerate}
\item $T^-\in \mathcal{E}_T$
\item $\mathcal{E}_T$ is closed under (countable) unions
\item $\mathcal{E}_T$ is closed under complementation
\item $P_T(T^-)=1$
\item $P_T$ is finitely additive over pairwise disjoint events.
\end{enumerate}
\end{teo}

\begin{proof}
1) $T^-$ is definable by the formula $X=x \lor X\neq x$, for some $X\in Dom$ and $x\in Ran(X)$. Therefore $T^-\in \mathcal{E}_T$.

2) If $S_1, S_2\in \mathcal{E}_T$, then $S_1$, resp. $S_2$ is definable by some formula $\alpha_1$, resp. $\alpha_2$. Then $S_1 \cup S_2$ can be defined using $\alpha_1\lor \alpha_2$.  
Thus  $S_1 \cup S_2\in \mathcal{E}_T$.

3) 
Let $S\in\mathcal{E}_T$. Then $S$ is defined by some $\CO$ formula $\alpha$. But then, since $\neg$ behaves as ordinary negation on single assignments (i.e. $(\{s\},\F)\models \neg\alpha$ iff $(\{s\},\F)\not\models \alpha$), 
 $T^-\setminus S$ is defined by $\neg\alpha$. Therefore $T^-\setminus S \in \mathcal{E}_T$.




4) $P_T(T^-) = P_T(X=x \lor X\neq x)=1$, since all the singleton subteams of $T$ satisfy $X=x \lor X\neq x$. 

5) Let $S_1,S_2 \in \mathcal{E}_T$, with $S_1 \cap S_2 = \emptyset$. $S_1$ and $S_2$ are defined by two $\CO$ formulas $\alpha_1,\alpha_2$, respectively. Then, by the proof of 2., $S_1\cup S_2$ is defined by $\alpha_1 \lor \alpha_2$. So 
\[
P_T(S_1\cup S_2) = P_T(\alpha_1 \lor \alpha_2) = \frac{|\{s\in T^-  \ | \  \{s\}\models \alpha_1\lor \alpha_2\}|}{|T^-|} =
\]
\[
= \frac{|\{s\in T^-  \ | \  \{s\}\models \alpha_1\}| + |\{s\in T^-  \ | \  \{s\}\models \alpha_2\}|}{|T^-|} =
\]
\[
 =\frac{|\{s\in T^-  \ | \  \{s\}\models \alpha_1\}|}{|T^-|} + \frac{|\{s\in T^-  \ | \  \{s\}\models \alpha_2\}|}{|T^-|} =
\]
\[
 = P_T(\alpha_1) + P_T(\alpha_2) = P_T(S_1) + P_T(S_2);
\]
 the third equality holds due to the assumption $S_1 \cap S_2 = \emptyset$.
\end{proof}


\subsection{A probabilistic language} \label{sec: PCO language}


In order to allow our formal languages to talk about probabilities, first of all we allow for new types of atomic formulas. Let $\sigma=(Dom, Ran)$ be a signature.
\begin{df}
The set of \textbf{probabilistic atoms} of signature $\sigma$ is given by:
\[
\ \Pr(\alpha) \geq \epsilon \ | \ \Pr(\alpha) > \epsilon \ | \ \Pr(\alpha) \geq \Pr(\beta) \ | \ \Pr(\alpha) > \Pr(\beta)
\]
where $\alpha,\beta$ are formulas of $\CO$ and $\epsilon \in [0,1]\cap \mathbb{Q}$. The first two kinds will be called more  specifically \textbf{evaluation atoms}, and the other two \textbf{comparison atoms}. 
 Literals (as defined in Def. \ref{def: CO language}) and probabilistic atoms 
  will be called \textbf{atomic formulas}.

The \emph{probabilistic causal language} $\PCO_{\sigma}$ is given by the following clauses:
\[
\varphi::= \eta \mid \varphi \land \varphi \mid  \varphi \sqcup \varphi \mid \alpha\supset \varphi \mid \SET X = \SET x\cf \varphi
\]
where $\SET X\subseteq Dom$, $\SET x \in Ran(\SET X)$, $\eta$ is an atomic formula 
 and $\alpha$ a $\CO$ formula.
\end{df}

\noindent Notice that the antecedents of $\supset$ and the arguments of probability operators are $\CO$ formulas (and thus they may contain occurrences of $\lor$, but not of $\sqcup$ or $\Pr$). Outside of these contexts, $\sqcup$ is allowed while $\lor$ is not.


The semantics of many of the operators has already been introduced in section \ref{sec: level of events}, while discussing language $\CO$. We need further semantic clauses for $\sqcup$ (\emph{global disjunction}) and for the probabilistic atoms. 
 For what concerns the former:
\begin{itemize}
\item $T\models \psi\sqcup \chi$ if $T\models \psi$ or $T\models \chi$.
\end{itemize}





\noindent The semantics of the probabilistic atoms is straightforward, save for the convention that empty\footnote{Remember that a causal multiteam is said to be empty if such is its multiteam component.} causal multiteams satisfy all probabilistic atoms; this convention allows us to extend to $\PCO$ the empty team property, maintaining some continuity with the literature on team semantics.\footnote{The empty multiteam property also has the surprising consequence of guaranteeing the definability of material implication in $\PCO$, as will be discussed later in this section.} The semantic clauses for probabilistic atoms are:
\begin{itemize}
\item $T\models \Pr(\alpha)\geq \epsilon  \iff T \text{ is empty or } P_T(\alpha)\geq \epsilon$
\item $T\models \Pr(\alpha) > \epsilon  \iff T \text{ is empty or }  P_T(\alpha) > \epsilon$
\item $T\models \Pr(\alpha)\geq \Pr(\beta) \iff  T \text{ is empty or }  P_T(\alpha)\geq P_T(\beta)$
\item $T\models \Pr(\alpha) > \Pr(\beta)  \iff T \text{ is empty or }  P_T(\alpha) >  P_T(\beta)$
\end{itemize}

If $\Gamma\cup\{\varphi\}\subseteq \PCO_\sigma$, we will write $\Gamma\models_\sigma\varphi$ if all causal multiteams of signature $\sigma$ that satisfy $\Gamma$ also satisfy $\varphi$. We omit the subscript $\sigma$ if the signature is clear from the context. We write $\psi\models \varphi$ for $\{\psi\}\models\varphi$. Finally, if both $\psi\models \varphi$ and $\varphi\models \psi$, we write $\psi\equiv\varphi$.

The fact that the dual negation is definable for $\CO$ formulas grants us with the possibility of introducing many useful abbreviations in $\PCO$:

\begin{itemize}
\item $\Pr(\alpha) \leq  \epsilon$ for $\Pr(\neg\alpha) \geq 1 - \epsilon$
\item $\Pr(\alpha) < \epsilon$ for $\Pr(\neg\alpha) > 1 - \epsilon$
\item $\Pr(\alpha) = \epsilon$ for $\Pr(\alpha) \geq  \epsilon \land \Pr(\alpha) \leq  \epsilon$
\item $\Pr(\alpha) \neq  \epsilon$ for $\Pr(\alpha) >  \epsilon \sqcup  \Pr(\alpha) <  \epsilon$
\item $\bot$ for $X = x\cf X \neq x$
\item $\top$ for $X = x \cf X = x$
\end{itemize}
whose semantics is as expected (but true by default on empty causal multiteams).

Now we can go back to our initial example (PC). We will take it to be represented in $\PCO$ by the formula
\[
(Dies = 1 \land Treated = 1)\supset (Treated=0 \cf \Pr(Dies=0) = \epsilon).
\]
In section \ref{sec: mixed statements} we will give precise reasons for the correctness of this representation.

The two kinds of disjunction $\lor$ and $\sqcup$ play different roles within our formalism. Consider the example of tossing two fair coins, as depicted in the following table.
\begin{center}
\begin{tabular}{|c|c|c|}
\hline
\multicolumn{3}{|l|}{Key \ \ \ \ \! X \quad \quad \ \ \! Y}\\
\hline
\phantom{a}0\phantom{a} & tails & tails\\
\hline
1 & tails & heads \\
\hline
2 & heads  & tails \\
\hline
3 & heads & heads \\
\hline
\end{tabular}
\end{center}
\noindent The tensor disjunction $\lor$ can be used to express disjunctive events; so for example we will say that the disjunctive event ``$X=heads$ or $Y=tails$'' has probability $\frac{3}{4}$ by means of the formula $\Pr(X=heads \lor Y=tails)=\frac{3}{4}$. It can be checked that this statement is true in $T$. On the other hand, disjunctive statements about probabilities, such as ``either $X = heads$ has probability $\frac{1}{2}$  or $Y=tails$ has probaility $\frac{1}{2}$'' are expressed by means of the global disjunction: $\Pr(X=heads) =\frac{1}{2} \sqcup \Pr(Y=tails) = \frac{1}{2}$. This statement is also true in $T$.

The presence of the global disjunction $\sqcup$ entails that the language $\PCD$ is not flat. The presence of probabilistic atoms makes so that even the weaker property of \emph{downward closure} is not respected (i.e. it may happen that a formula is satisfied by a given causal multiteam, but not by some of its causal sub-multiteams). For a counterexample, in the causal  multiteam considered above (call it $T$) we have $\Pr(X = tails)\leq \frac{1}{2}$, but the causal sub-multiteam $T^{X = tails}$ constituted \emph{only} of the assignments that satisfy $X = tails$ will not satisfy $\Pr(X= tails)\leq \frac{1}{2}$.

In  \cite{BarSan2020} we emphasized the possibility of expressing dependencies as atomic formulas in causal team semantics. In particular, we considered an extension of $\CO$ that allows for functional dependence atoms $\dep{\SET X}{Y}$ that express that the values of $Y$  are functionally determined by the values of the tuple $\SET X$. Such atoms are a traditional object of study in database theory (\cite{AbiHulVia1995}) and have been introduced in the context of logic in \cite{Vaa2007}. In detail, we say that a causal (multi)team $T$ satisfies a dependence atom according to the following clause: 
\begin{itemize}
\item $T\models \dep{\SET X}{Y}$ iff for all $s,s'\in T^-$, $s(\SET X) = s'(\SET X)$ implies $s(Y) = s'(Y)$.
\end{itemize}

\noindent The reason why we do not consider such extensions here is that the functional dependence atoms are already definable in $\PCO$. It is easy to show that, for any causal multiteam $T$ of appropriate signature,
\[
T\models\dep{\SET X}{Y} \text{ if and only if } T\models \bigwedge_{\SET x\in Ran(\SET X)} \left(\SET X = \SET x \supset \bigsqcup_{y \in Ran(Y)} Y=y\right).  
\] 


\subsection{Definable connectives}

\commf{If we want to shorten the paper, we could omit the rest of this section, since contradictory negation and material implication are not used in the rest of the paper. (They are important instead for the proof theory.)}
Perhaps surprisingly, in $\PCO$ we can inductively define a \emph{weak contradictory negation}, i.e. a negation that behaves like classical negation except on empty causal multiteams. We give clauses for both the basic syntax and for the defined abbreviations:

\begin{itemize}
\item $(\Pr(\alpha)\geq \epsilon)^C$ is $\Pr(\alpha) < \epsilon$ (and vice versa)  
\item $(\Pr(\alpha) > \epsilon)^C$ is $\Pr(\alpha)\leq \epsilon$ (and vice versa)
\item $(\Pr(\alpha) = \epsilon)^C$ is $\Pr(\alpha)\neq \epsilon$ (and vice versa)
\item $(\bot)^C$ is $\top$ (and vice versa)
\item $(X=x)^C$ is $\Pr(X=x) <1$
\item $(X\neq x)^C$ is $\Pr(X\neq x) <1$
\item $(\psi \land \chi)^C$ is $\psi^C \sqcup \chi^C$
\item $(\psi \sqcup \chi)^C$ is $\psi^C \land \chi^C$
\item $(\alpha \supset \chi)^C$ is $\Pr(\alpha)>0 \land \alpha \supset \chi^C$
\item $(\SET X = \SET x \cf \chi)^C$ is $\SET X = \SET x \cf \chi^C$.
\end{itemize}

\noindent In the clause for $\supset$, the conjunct $\Pr(\alpha)>0$  (whose intuitive interpretation is ``if $T$ is nonempty, then $T^\alpha$ is nonempty'') is added to insure that $(\alpha \supset \chi)^C$ is not satisfied by $T$ in case ($T$ is nonempty and) $T^\alpha$ is empty.\footnote{We remark that $\Pr(\alpha)>0$ could also be replaced with $(\neg\alpha)^C$.}\footnote{For a comparison, notice that while $(\alpha \supset \chi)^C$ is $\Pr(\alpha)>0 \land \alpha \supset \chi^C$, we have $\neg(\alpha \supset \beta) \equiv \alpha \land \neg \beta$.}

The meaning of the weak contradictory negation is as follows.

\begin{teo}\label{thm: contradictory negation}
For any $\varphi$ in $\PCO_{\sigma}$ and any \emph{nonempty} causal multiteam $T$ of signature $\sigma$,
\[
T\models \varphi^C \iff T\not\models \varphi.
\]
\end{teo}

\begin{proof}
We prove the statement simultaneously on all causal multiteams by induction on 
 $\varphi$. We show the case for $\supset$, which is the only one that is not immediate.


Now suppose $T = (T^-,\F)\models 
\Pr(\alpha)>1  \land \alpha \supset \chi^C$. Thus $T^\alpha\models \chi^C$. 
Since $T$ is nonempty and $T\models \Pr(\alpha)>0$, there is an $s\in T^-$ such that $(\{s\},\F)\models \alpha$: i.e., $T^\alpha$ is nonempty. Then we can apply the inductive hypothesis on $\chi$ and obtain $T^\alpha\not \models \chi$. Thus, $T\not\models \alpha\supset\chi$.

Viceversa, assume $T\not\models \alpha\supset\chi$. Then $T^\alpha\not \models \chi$, which in particular (by the empty team property) entails that $T^\alpha$ is not empty. Thus we can apply the i.h. to $\chi$ and obtain $T^\alpha \models \chi^C$; thus $T\models \alpha \supset \chi^C$. Furthermore, since $T^\alpha$ is not empty, 
$T\models \Pr(\alpha)>0$.
\end{proof}

\noindent It is not difficult to prove that the two kinds of negation are related by the following law: whenever $\alpha$ is a $\CO$ formula without occurrences of $\lor$, we have $\neg\alpha \models \alpha^C$. 

Since we have at our disposal a negation that behaves almost like the usual contradictory negation, we can  introduce a conditional which behaves as the material conditional:
\begin{itemize}
\item $\psi\rightarrow \chi$ stands for $\psi^C \sqcup \chi$.
\end{itemize}
We then have $T\models \psi\rightarrow \chi$ iff $T\not\models\psi$ or $T\models\chi$ (notice that if $T$ is empty, by the empty multiteam property of $\PCO$ we have $T\models\chi$ and thus $T\models\psi^C \sqcup \chi$, $T\models \psi\rightarrow \chi$).

We end this section with a brief comparison of $\rightarrow$ and $\supset$. 
\corrf{Note that $\alpha\rightarrow\beta$  
and  $\alpha\supset\beta$ are not in general equivalent even if $\alpha,\beta$ are $\CO$ formulas. Consider for example a causal multiteam $T$ with two assignments $s= \{(X,0),(Y,0)\}$ and $t= \{(X,1),(Y,1)\}$. Clearly $T\models X=0 \rightarrow Y=1$ (since $T\not\models X=0$), while $T\not\models X=0 \supset Y=1$ (since $T^{X=0}\not\models Y=1$). The opposite direction holds (when the syntax allows to formulate it):
}


\begin{prop}
Let $\alpha\in\CO,\psi\in\PCO$.
\begin{enumerate}
\item $\alpha\supset\psi \models \Pr(\alpha)=1\rightarrow \psi$.
\item If $\alpha$ has no occurrences of $\lor$, $\alpha\supset\psi \models \alpha \rightarrow \psi$.
\end{enumerate}
\end{prop}

\begin{proof}
We prove 1. (the other proof is analogous). Assume $T\models \alpha\supset\psi$; then $T^\alpha\models \psi$. If now $T\models \Pr(\alpha)=1$, we have that $T = T^\alpha\models \psi$.
\end{proof}

\section{Comparison with causal models} \label{sec: comparison with causal models}

Probabilistic notions of causation have been extensively studied in
the literature. In the context of \emph{Bayesian networks} (graphs enriched
with a joint probability distribution over the nodes of the graph) some conditional probabilities are interpreted as causal relations (see Spirtes, Glymour and Scheines \cite{SpiGlySch1993}). 
In the approach of \cite{Pea2000} and \cite{SpiGlySch1993}, by a \emph{causal model}, or \emph{semi-deterministic structural equation model}, one usually understands (the terminology is by no means used consistently) a system of structural equations paired with a probabilistic distribution over the set of exogenous variables $\SET{U}$. The \emph{recursive} semi-deterministic structural equation models can be seen as a special case of Bayesian networks: the case in which the joint probability distribution over the variables of the graph is generated by a probability distribution over the exogenous variables, together with a system of structural equations for the endogenous variables. 
 Indeed, if the causal model is recursive (i.e. the graph is acyclic), then each endogenous variable $V$ can be represented as a function of the exogenous variables, and its probability distribution can be computed from the joint probability distibution of the exogenous variables. In order to do this, first observe that the causal law for $V$, which we have been encoding by a function $\F_V$ of all the remaining variables, can also be expressed by a function $f_V$ without dummy arguments, i.e. a function whose only arguments are the variables in $PA_V$ (such function is defined by the conditions $f_V(\SET p)\dfn\F_V(\SET p,\SET d)$, where $\SET d$ is an arbitrary list of values for the dummy arguments $Dom \setminus PA_V$). Let us call the expression $V:=f_V(PA_V)$ the \emph{minimal} structural equation for $V$. Next, we can iteratively replace each of the endogenous variables occurring in the right-hand side of $V:=f_V(PA_V)$ with the right-hand side of their own minimal structural equations, until we obtain a function $f^*_V(\SET{U})$ (the procedure will terminate by recursivity). Similarly, given recursivity, the probability distribution over exogenous variables induces a joint probability distribution over \emph{all} the variables of the system given by the following equation:
\begin{equation}\label{equation: extending probabilities}
P(X_1 = x_1\land\dots\land X_n = x_n) \dfn
\end{equation}
\[
 \sum_{\{(u_1,\dots,u_k)  \ | \ \text{for all }i=1..n \text{ : } f^*_{X_i}(u_1,\dots,u_k)=x_i\}} P(U_1=u_1\land\dots\land U_k=u_k)
\]
\corrf{where $U_1, \dots, U_k$ is a list of the exogenous variables.}
From this joint probability, one can obtain by marginalization the distributions over each set of variables. 
 
In causal multiteams probabilities are defined from a multiteam by counting, as explained in section \ref{subs: probabilities}.    It can be shown 
 that the probabilities that are induced by recursive causal multiteams respect the structure of the probabilities typical of recursive causal models. That is, if we take the team-defined probability of the exogenous variables $\SET U$ in $T$ \corrf{(what we called $P_T(\SET U = \SET u)$)}, and we define from it the joint probability $P(\SET X)$ over $Dom(T)$ according to equation \eqref{equation: extending probabilities}, we obtain exactly the team-defined probability distribution of $\SET X$:
\[
P(\SET X=\SET x) = P_T(\SET X=\SET x)
\]
for all tuples $\SET x \in Ran(\SET X)$. This is a direct consequence of the fact that the assigments in a causal multiteam are compatible with the causal laws, which entails that we have $|\{s\in T^- \mid s(X_1)=x_1 \land \dots \land s(X_n)=x_n\}| = |\{s\in T^- \mid \text{for all $i=1,\dots,n$, } f^*_{X_i}(s(U_1),\dots,s(U_k)) = x_i\}|$.

Another way of interpreting this result is the following: each recursive causal multiteam induces a corresponding semi-deterministic recursive SEM. It is straightforward to see that a sort of converse is also true: each recursive, semi-deterministic SEM can be mimicked by an appropriate causal  multiteam $T$.\footnote{Actually, for each SEM there is a whole class of multiteams that are adequate for the task. See the notion of \emph{rescaling} in section \ref{sec: infinitary language}.} It is immediate to obtain the function component 
 $\F$ from the structural equations of the SEM; we are left with the task of defining a multiteam $T^-$, compatible with 
 $\F$, which induces the same probability distribution as the SEM. But this is easy; let $P$ be the probability distribution over the exogenous variables of the SEM; let $P(u_1^i,\dots,u_k^i)=\frac{a_i}{b}$ for each tuple $(u_1^i,\dots,u_k^i)$ of values for the exogenous variables $\SET U = \{U_1,\dots,U_k\}$ of the SEM; $b$ is a common denominator. Then, for each tuple $(u_1^i,\dots,u_k^i)$, we require $T^-$ to contain exactly $a_i$ copies of the assignment $((U_1,u_1^i),\dots,(U_k,u_k^i),(V_1, f'_{V_1}(s(\SET U))),\dots,$ $(V_m, f'_{V_m}(s(\SET U))))$, where $V_1,\dots,V_m$ is a list of the endogenous variables. It is immediate to verify that the probability distribution over $\SET U$ induced by the multiteam coincides with the probability given by the SEM; and, by  
the earlier considerations of this subsection, the same can be said of the probabilities assigned to any set of variables in the domain.

The approach encoded in \eqref{equation: extending probabilities} fails, however, for non-recursive causal models, 
in case their system of structural equations has multiple solutions or no solutions at all for some configuration(s) of values of the exogenous variables. In this case some endogenous variables may fail to be functions of the exogenous variables. Consider e.g. a system with variables $U,X,Y$ all having range $\{0,1\}$, and consisting of the equations $X:=Y$ and $Y:= X$, which induce the arrows $X\rightarrow Y$ and $Y\rightarrow X$; here $U$ is the only exogenous variable; given a fixed value $u$ for $U$, both triples $(u,0,0)$ and $(u,1,1)$ are solutions to the system; therefore, neither $X$ nor $Y$ can be determined as functions of $U$. 

Unlike the semi-deterministic approach of causal models, our definition of probabilities for causal teams can in principle be applied also to non-recursive systems, and is thus more general. 
 What is still lacking in the nonrecursive case, with our semantics, is a mechanism for computing the new post-intervention probabilities. When one assignment produces multiple solutions, we should redistribute the probabilistic weight of that assignment over multiple solutions, but the way this should be done is plausibly determined by the specific intended application of the formalism. These complications go beyond the scope of the present paper.

As a final remark, we notice that often, in the literature on Bayesian networks, not all combinations of graphs and probabilities are allowed; it is instead required that the graph-probability pair satisfy the so-called \emph{Markov condition}. This condition amounts to the requirement that each variable be probabilistically independent of its non-descendants (non-effects) conditional on its parent set; it is considered by many authors to be a necessary condition for giving a causal interpretation to Bayesian networks (see e.g. \cite{HauWoo1999} for a discussion). It is well known (see again \cite{HauWoo1999}) that, in the special case of \emph{recursive} SEMs, the Markov Condition is trivially satisfied by the endogenous variables, and it ultimately reduces to the following constraint on pairs of exogenous variables:
\[
U_i \pindep U_j\text{, for all exogenous variables } U_i \neq U_j
\]
where  $\pindep$ denotes probabilistic independence. It is immediate to see that our causal multiteams need not satisfy this condition over their set of exogenous variables. We might call \textbf{Markovian} the causal multiteams that satisfy this condition; we expect Markovian causal multiteams to be the most adequate tool for applications of our framework, but we will not investigate further the consequences of the Markov condition in this paper.

\section{Definable concepts}\label{sec: definable concepts}

In this subsection we will see that, although our proposed language $\PCO$ is relatively poor, it already allows one to discuss easily some concepts that are inaccessible or at least rather murky when using the typical notations of causal inference. The key to this kind of expressivity is having two kinds of conditional in the language, $\cf$ and $\supset$. We begin by considering to what extent conditional probabilities are definable in $\PCO$.

\subsection{Conditional probabilities} \label{subs: conditional probabilities of CO}


 The obvious way of applying the definition of a conditional probability to events described in a (nonempty) causal multiteam is (assuming that $\alpha,\gamma$ are $\CO$ formulas, and $\gamma$ is satisfied by at least one assignment of the team) by the definition
\[
P_T(\alpha \mid \gamma)\dfn \frac{P_T(\alpha\land\gamma)}{P_T(\gamma)}.
\]
How can we express statements on $P_T(\alpha  \mid \gamma)$ in our language $\PCD$? The idea we shall follow is that a conditional probability can be seen as a probability over a restricted sample space, and  selective implication provides this restriction. We should then expect atoms of the form $\Pr(\alpha  \mid  \gamma)\geq \epsilon$ to be equivalent with 
\[
\gamma \supset \Pr(\alpha)\geq \epsilon
\]
and, similarly, the comparison atoms $\Pr(\alpha  \mid  \gamma)\leq \Pr(\beta  \mid  \gamma)$, with both sides conditioned over a common formula $\gamma$, to be equivalent to
\[
\gamma \supset \Pr(\alpha)\geq \Pr(\beta).
\] 
We can prove that this is indeed the case, once we give a precise meaning to the atoms featuring conditional probabilities (which might be thought, say, as additional operators in an extension $\PCO^*$ of the language $\PCO$). Write $\vartriangleright$ for either $\geq$ or $>$. Taking into account that $T^- =\emptyset$ implies $(T^\gamma)^- = \emptyset$, a moment of reflection suggests the following semantic clauses for conditional probabilistic atoms:
\begin{align*}
T\models & \Pr(\alpha \mid \gamma) \vartriangleright \epsilon \text{ iff }  (T^\gamma)^- = \emptyset \text{ or } \frac{P_T(\alpha\land\gamma)}{P_T(\gamma)}\vartriangleright \epsilon.\\
T\models & \Pr(\alpha \mid \gamma) \vartriangleright \Pr(\beta \mid \delta) \text{ iff } (T^\gamma)^-= \emptyset \text{ or } (T^\delta)^- = \emptyset \text{ or } \frac{P_T(\alpha\land\gamma)}{P_T(\gamma)}\vartriangleright \frac{P_T(\beta\land\delta)}{P_T(\delta)}.
\end{align*}
We then have

\begin{teo}
For every causal multiteam $T$, all $\epsilon\in [0,1]\cap\mathbb Q$ and all $\CO$ formulas $\alpha,\beta,\gamma$,
\begin{enumerate}
\item $T\models \Pr(\alpha  \mid  \gamma)\vartriangleright \epsilon      \iff    T\models  \gamma \supset \Pr(\alpha)\vartriangleright \epsilon$
\item $T\models \Pr(\alpha  \mid  \gamma)\vartriangleright \Pr(\beta  \mid  \gamma)     \iff      T\models \gamma \supset \Pr(\alpha)\vartriangleright \Pr(\beta)$.
\end{enumerate}
 
\end{teo}

\begin{proof}
We prove only 1., as the calculations for 2. are completely analogous. We have: 
\begin{align*}
T\models \gamma \supset \Pr(\alpha)\vartriangleright \epsilon & \iff   T^{\gamma}\models \Pr(\alpha)\vartriangleright \epsilon \\
 & \iff \frac{|\{s\in (T^{\gamma})^-  \ | \ (\{s\},\F)\models \alpha\}|}{|(T^{\gamma})^-|}\vartriangleright \epsilon \\
 & \iff \frac{|\{s\in (T^{\gamma})^-  \ | \  (\{s\},\F)\models \alpha\}|}{|T^-|} \cdot \frac{|T^-|}{|(T^{\gamma})^-|} \vartriangleright \epsilon \\
 & \iff \frac{P_T(\alpha \land \gamma)}{P_T(\gamma)} \vartriangleright \epsilon \\ 
 & \iff P_T(\alpha \ | \ \gamma) \vartriangleright \epsilon.
\end{align*}
\end{proof}

\noindent The case of comparison atoms of the general form $\Pr(\alpha \mid \gamma) \vartriangleright \Pr(\beta \mid \delta)$ -- with $\gamma$ possibly distinct from $\delta$ -- is more subtle, and it is not immediately evident whether such atoms may be definable or not in $\PCO$. The forthcoming \cite{BarVir2023} shows that they are not, and that there are some specific instances that are not expressible at all (more precisely, even simpler atoms of the form $\Pr(\alpha \mid \gamma) \vartriangleright \Pr(\beta)$, aka $\Pr(\alpha \mid \gamma) \vartriangleright \Pr(\beta \mid \top)$, are not in general expressible in $\PCO$). The method of proof, involving an abstract characterization of $\PCO$ and some geometry, goes beyond the scope of the present paper.

 Despite this limitation in the expressibility of conditional probabilities, we see, starting from the next subsection, that the selective implication operator has a wide range of applications within the context of $\PCO$. Its wide applicability stems from the fact that $\supset$ \corrf{distributes over} 
  all the other operators (except $\cf$) and that $\cf$ distributes over $\supset$.  

We remark that the general comparison atoms would be definable if we allowed for either existential quantifiers or infinitary disjunctions, as follows:
\[
T\models \Pr(\alpha \mid \gamma) \vartriangleright \Pr(\beta \mid \delta) \iff T\models \bigsqcup_{\epsilon\in [0,1]\cap\mathbb Q} \big((\gamma \supset \Pr(\alpha) \vartriangleright \epsilon) \land  (\delta \supset\Pr(\beta)\vartriangleleft \epsilon)\big).
\]
where $\vartriangleleft$ denotes the converse of $\vartriangleright$. 
We will consider such an infinitary language in section \ref{sec: infinitary language}.

A related issue is whether \emph{probabilistic conditional independence} is definable or at least expressible in our languages.
 An expression of the form  $\alpha\pindep\beta$ states that $\alpha$ and $\beta$ are probabilistically independent; we can take it to be an abbreviation for $\Pr(\alpha) = 0\sqcup\Pr(\beta \mid \alpha) = \Pr(\beta)$. 
  More generally, we can talk of conditional independence atoms $\alpha\pindep_{\gamma}\beta$ as abbreviations for 
\[
 \Pr(\gamma) = 0 \ \sqcup \ \Pr(\alpha \land \gamma) = 0 \   \sqcup \  \Pr(\beta \ | \ \alpha\land\gamma) = \Pr(\beta \ | \  \gamma).
\]
While conditional comparison atoms of the general form  
 $\Pr(\alpha \mid \beta)=\Pr(\gamma\mid \delta)$ are not in general expressible in $\PCO$ (\cite{BarVir2023}), 
 we leave it as an open issue whether the specific atom $\Pr(\beta \ | \ \alpha\land\gamma) = \Pr(\beta \ | \  \gamma)$ needed here is expressible, and whether the conditional independence atoms are. The results in \cite{HanVir2022} point towards inexpressibility of these atoms, but the methods that are used in that paper rely on quantifiers.

\subsection{Mixed statements} \label{sec: mixed statements}

The core of structural equation modeling is the study of  expressions which involve both classical probabilistic operations, such as conditioning, and causal operations. For example, in \cite{PeaGlyJew2016} we find expressions like:
\begin{equation}\label{equation: postintervention conditioning}
P(\SET Y=\SET y \ | \ do(\SET X=\SET x),\SET Z=\SET z)=\epsilon
\end{equation}
\begin{equation}\label{equation: preintervention conditioning}
P(\SET Y_{\SET X=\SET x}=\SET y \ | \ \SET Z=\SET z)=\epsilon,
\end{equation}
some instances of which we have considered in the introduction. 
Both notations describe probabilities that subsist in a given SEM after an intervention $do(\SET X = \SET x)$. The difference between the two cases is that, in the former, the probability is conditional upon $\SET Z$ being equal to $\SET z$ \emph{in the post-intervention system}; while the latter refers to $\SET Z$ being equal to $\SET z$ in the \emph{initial} (pre-intervention) system. We have here two distinct notations in which the condition $\SET Z=\SET z$ appears essentially in the same position, but with different meanings. Thus this notation is misleading; furthermore, it is not at all perspicuous about the order in which observation (conditioning) and intervention take place; it does not even suggest the fact the two operations do not commute with each other. Indeed, an expression of the form $P(\SET Y=\SET y \ | \ \SET Z=\SET z , do(\SET X=\SET x))=\epsilon$ is simply taken as synonym of \eqref{equation: postintervention conditioning}.

 It has been observed 
 that formulas of the form \eqref{equation: preintervention conditioning} can assign nonzero probability to expressions that may superficially appear to be inconsistent. One such example, taken from \cite{PeaGlyJew2016}, is the statement that ``the probability that the duration of my journey would have been $y$ if I had taken the highway, $X=1$, when in fact I have not, $X=0$'',  which is rendered as $P(Y_{X = 1}= y \ | \ X=0)$.\footnote{Of course, the ``inconsistency'' of $X=0$ with $X=1$ is solved by the fact that the two formulas are evaluated in two different SEMs, or two different causal teams -- say $T$ and $T_{X=1}$, respectively.} This is one of the reasons why, in the literature, sometimes only expressions of the form \eqref{equation: preintervention conditioning} are called counterfactuals, while expressions of the form \eqref{equation: postintervention conditioning} are simply labelled as $do$ expressions. 

The language $\PCD$, instead, does not adopt special notations to distinguish the two cases; the distinction is simply given by a different order in the application of the two conditionals $\supset$ and $\cf$. Indeed, \eqref{equation: postintervention conditioning} can be expressed in $\PCD$ as
\[
\SET X=\SET x \cf (\SET Z= \SET z \supset \Pr(\SET Y= \SET y) = \epsilon)
\] 
while \eqref{equation: preintervention conditioning} corresponds to the formula
\[
 \SET Z= \SET z \supset ( \SET X=\SET x \cf \Pr(\SET Y= \SET y) = \epsilon).
\] 

\noindent This notation now makes it completely clear that two different kinds of conditioning are involved, and that observation and intervention are performed in a certain order, thus avoiding the ambiguities of \eqref{equation: postintervention conditioning} and \eqref{equation: preintervention conditioning}. Simple examples can be devised to show that, indeed, these two expressions are not equivalent.\footnote{E.g., the examples at the end of section 5.1 of \cite{BarSan2020} can be adapted to the probabilistic context.} We now prove that they have the intended meaning.

\begin{teo} Let $\sigma=(Dom,Ran)$, $\alpha,\beta$ be $\CO_{\sigma}$ formulas, $T=(T^-,\F)$ a nonempty, recursive 
 causal multiteam, $\SET X \in Dom$. Then:
\begin{enumerate}
\item $T\models \SET X=\SET x \cf (\alpha \supset \Pr(\beta) = \epsilon)$ if and only if $P_{T_{\SET X = \SET x}}(\beta \mid \alpha) = \epsilon$ \corrf{or $(T_{\SET X = \SET x})^\alpha$ is empty}.
\item $T\models \alpha \supset (\SET X=\SET x \cf \Pr(\beta) = \epsilon)$ if and only if $P_T(\SET X=\SET x \cf \beta \mid \alpha) = \epsilon$ \corrf{or $T^\alpha$ is empty}.
\end{enumerate}
Similar statements hold with $\leq,<,\geq$ or $>$ in place of $=$.
\end{teo}

\noindent The theorem makes it clear that the former is a conditional probability in the intervened multiteam, while the latter is a conditional probability in the initial causal multiteam.

\begin{proof}
1) We prove the statement only for $=$, as the other cases are analogous. By the semantical clauses, $T\models \SET X=\SET x \cf (\alpha \supset \Pr(\beta) = \epsilon)$ if and only  $((T_{\SET X=\SET x})^\alpha)^-$ is empty or $(T_{\SET X=\SET x})^\alpha\models \Pr(\beta) = \epsilon$. Under the assumption of nonemptiness of $((T_{\SET X=\SET x})^\alpha)^-$, the following equivalences hold:
\[
(T_{\SET X=\SET x})^\alpha\models \Pr(\beta) = \epsilon  \iff
\]
\[
\iff \frac{|\{s\in((T_{\SET X=\SET x})^\alpha)^- \ | \ \{s\}\models \beta\}|}{((T_{\SET X=\SET x})^\alpha)^-}  = \epsilon
\]
\[
\iff \frac{|T_{\SET X = \SET x}^-|}{|((T_{\SET X=\SET x})^\alpha)^-|}  \cdot  \frac{|\{s\in((T_{\SET X=\SET x})^\alpha)^- \ | \ \{s\}\models \beta\}|}{|T_{\SET X = \SET x}^-|}  = \epsilon
\]
\[
\iff \frac{|T_{\SET X = \SET x}^-|}{|((T_{\SET X=\SET x})^\alpha)^-|}  \cdot  \frac{|\{s\in T_{\SET X = \SET x}^- \ | \ \{s\}\models \beta \land \alpha\}|}{|T_{\SET X = \SET x}^-|}  = \epsilon
\]
\[
\iff \frac{1}{P_{T_{\SET X = \SET x}}(\alpha)} \cdot P_{T_{\SET X = \SET x}}(\beta\land \alpha) = \epsilon
\]
\[
\iff P_{T_{\SET X = \SET x}}(\beta \ | \ \alpha) = \epsilon.
\]

2) Notice that $(T^{\alpha})_{\SET X=\SET x}^-$ is nonempty if and only if  $(T^\alpha)^-$ is nonempty, since interventions preserve the size of multiteams. Then, assuming 
 $(T^\alpha)^-$ is nonempty, the equivalence can be proved as follows:
\[
T\models \alpha \supset (\SET X=\SET x \cf \Pr(\beta) = \epsilon) \iff
\]
\[
\iff (T^{\alpha})_{\SET X=\SET x} \models \Pr(\beta) = \epsilon
\]
\[
\iff \frac{|\{s\in (T^{\alpha})_{\SET X=\SET x}^- \ | \ (\{s\},\F_{\SET X=\SET x})\models \beta \}|}{|(T^{\alpha})_{\SET X=\SET x}^-|} = \epsilon
\]
\[
\iff \frac{|\{t\in (T^{\alpha})^- \ | \ (\{t\}^\F_{\SET X=\SET x},\F_{\SET X=\SET x})\models \beta\}|}{|(T^{\alpha})^-|} = \epsilon
\]
\[
\iff \frac{|\{t\in (T^{\alpha})^- \ | \ (\{t\}^\F_{\SET X=\SET x},\F_{\SET X=\SET x})\models \beta\}|}{|T^-|}\cdot\frac{|T^-|}{|(T^{\alpha})^-|} = \epsilon
\]
\[
\iff \frac{P_T(\alpha \land (\SET X=\SET x \cf \beta))}{P_T(\alpha)} = \epsilon
\]
\[
\iff P_T(\SET X=\SET x \cf \beta \ | \ \alpha) = \epsilon
\]
where the fourth equivalence makes again essential use of the fact that interventions preserve the size of multiteams.  
\end{proof}


It is also worth mentioning that, in some papers on deterministic causal models (e.g. \cite{GalPea1998}, \cite{Hal2000}) the notation $\SET Y_{\SET x}(\SET u)=\SET y$ is used to express that, under the \corrf{assumption} 
 that the (whole set of) exogenous variables $\SET U$ has values $\SET u$, fixing some endogenous variables $\SET X$ to $\SET x$ forces \corrf{the exogenous variables} $\SET Y$ to take value $\SET y$.\footnote{This notation, in turn, is adapted from the literature on the Neyman-Rubin ``potential outcome'' approach to causation (\cite{Ney1923},\cite{Rub1974}).} We could try to  extend the use of this notation to causal (multi)teams. One immediately wonders, at this point, whether the \corrf{assumption}  
  $\SET U= \SET u$ refers to the pre- or the post-intervention system.  Since exogenous variables are not affected by interventions over endogenous variables, in case $\SET X \cap \SET U = \emptyset$ the answer is that the statements $\SET U=\SET u \supset (\SET X=\SET x \cf \SET Y=\SET y)$ and $\SET X=\SET x \cf (\SET U=\SET u \supset \SET Y=\SET y)$ are equivalent; in this case, pre- and post-conditioning coincide. But if it happens that  $\SET X \cap \SET U \neq \emptyset$ -- as is allowed e.g. in our formalism -- then the two statements may not be equivalent, and the notation $\SET Y_{\SET x}(\SET u)=\SET y$ is not adequate anymore. Some of the authors who use this kind of notations tend indeed to ignore the possibility of interventions over exogenous variables, perhaps because of some form of awareness of the notational ambiguity that may arise.

\subsection{Normal form}\label{subs: normal form}

The language $\PCD$ allows for complex combinations of the conditionals $\supset$ and $\cf$, in ways that may seem at first impossible to capture by the notational schemata \eqref{equation: postintervention conditioning} and \eqref{equation: preintervention conditioning}. One example is provided by the case in which we try to assert that a probability is conditional over both a pre-intervention condition $\alpha$ and a post-intervention condition $\beta$: 
\begin{equation}\label{equation: prepostintervention conditioning}
 \alpha \supset (\SET X=\SET x \cf (\beta \supset \Pr(\psi) = \epsilon)).
\end{equation}
The discussion in \cite{PeaGlyJew2016} (section 4.1) shows that the need to evaluate these kinds of probabilities may arise in practical applications, and also shows (in a particular case) how to bend the notation template \eqref{equation: preintervention conditioning} to cover particular instances of this scheme; for example, in line with this suggestion, our $\PCO$ statement
\begin{equation} \label{equation: prepostintervention example}
\alpha \supset (\SET X=\SET x \cf (\SET Z=\SET z \supset \Pr(\SET Y = \SET y) = \epsilon))
\end{equation}
 might be rendered as:
\begin{equation}\label{equation: bent potential outcome notation}
P(\SET Y_{\SET X= \SET x} = \SET y  \ | \  \SET Z_{\SET X= \SET x} = \SET z, \alpha) = \epsilon 
\end{equation}
where the subscript of $\SET Z$ indicates that the condition $\SET Z = \SET z$ is intended as post-intervention. Indeed, similar expressions occur e.g. in \cite{PeaGlyJew2016}. 
 We point out that in $\PCD$ we can in general establish the equivalence of 
   $(\SET X = \SET x \cf \alpha) \supset (\SET X = \SET x \cf\psi)$ and $\SET X = \SET x \cf (\alpha \supset \psi)$, i.e. the distributivity of $\cf$ over $\supset$; therefore, \eqref{equation: prepostintervention example} might also be written as 
\[
\alpha\supset \big[ (\SET X=\SET x \cf \SET Z=\SET z) \supset (\SET X=\SET x \cf \Pr(\SET Y = \SET y) = \epsilon) \big]
\]
which provides some sort of justification for the ``equivalence'' between \eqref{equation: prepostintervention example} and \eqref{equation: bent potential outcome notation}. 


One could stretch the counterfactual notation \eqref{equation: bent potential outcome notation} in order to represent more complex forms of conditioning; the notations, however, become quickly quite complicated. As a still tolerably simple example of this kind, we may consider the case of conditioning over some condition $\alpha$ that holds in between two interventions:  
\begin{equation}\label{equation: infraintervention conditioning}
\SET X=\SET x \cf (\alpha \supset (\SET W = \SET w \cf \Pr(\psi) = \epsilon)).
\end{equation}
For example, one might want to express that $\epsilon$ is the probability that a patient that reacts badly ($\alpha$) to a given therapy ($\SET X=\SET x$) would recover ($\psi$) after receiving a second, distinct kind of treatment ($\SET W = \SET w$). How could we express such a probability in the informal spirit that animates the notation of \cite{PeaGlyJew2016}? We need to talk of a condition $\alpha$ that may occur after the intervention $do(\SET X = \SET x)$; and a condition $\psi$ that may occur after the \emph{sequence} of interventions $do(\SET X = \SET x), do(\SET W = \SET w)$. The trick that makes the notation \eqref{equation: bent potential outcome notation} work once more is the following:  any sequence of interventions has the same effect as an appropriate single intervention, the idea being that, if there is some variable which is affected by both interventions, then the effect of the second intervention prevails.\footnote{This issue is discussed in \cite{BarSan2018} (section 7) and \cite{BarSan2020} (section 4.3), in reference to what is called there the ``overwriting rule''. 
 Similar ideas underlie an inference rule described in \cite{Bri2012}.} 
In the case at hand, let $\SET X'$ stand for the set of variables of $\SET X$ that are not in $\SET W$, $\SET X' = \SET X \setminus \SET W$, and $\SET x'$ stand for the corresponding values from $\SET x$. Then the sequence of interventions $do(\SET X = \SET x), do(\SET W = \SET w)$ has the same effect as the single intervention $do(\SET X' = \SET x' \land \SET W = \SET w)$. Combining this observation with the distributivity of $\cf$ over $\supset$, 
 we can transform \eqref{equation: infraintervention conditioning} into 
\[
(\SET X=\SET x \cf \alpha) \supset \big[(\SET X' = \SET x' \land \SET W = \SET w) \cf \Pr(\psi) = \epsilon\big].
\]
 Writing $\alpha_{\SET X=\SET x}$ for the ``formula'' obtained from $\alpha$ by adding the subscript $_{\SET X=\SET x}$ to each variable occurring in it, we can then convert this last expression into a notation similar to those which are used in \cite{PeaGlyJew2016}:
\[
P(\SET Y_{\SET X'= \SET x' \land \SET W = \SET w} = \SET y  \ | \  \alpha_{\SET X= \SET x}) = \epsilon 
\]
 This example illustrates the increasing difficulty, within the formalism of \cite{PeaGlyJew2016}, in producing correct expressions for concepts that are very naturally expressed in $\PCO$.
Interestingly, however, this kind of notational translation can in principle always be achieved, as illustrated by the following normal form result.

\begin{teo}[Pearl-style normal form]\label{thm: Pearl normal form}
Every $\PCO_{\sigma}$ formula $\varphi$ is equivalent to a $\PCO_{\sigma}$ formula $\varphi'$ such that:
\begin{enumerate}[A.]
\item all consequents of $\cf$ are probabilistic atoms
\item all consequents of $\supset$ are counterfactuals or probabilistic atoms.
\end{enumerate}
\end{teo}

\noindent In other words, every $\PCO$ formula can be written as a Boolean combination (using $\land,\sqcup$) of three simple types of formulas (writing $t$ for either $\epsilon\in[0,1]\cap \mathbb Q$ or $\Pr(\beta)$):
\begin{enumerate}
\item conditional probability statements: $\gamma \supset \Pr(\alpha) \vartriangleright t$
\item $do$ expressions 
 without conditioning: $\SET X = \SET x \cf \Pr(\alpha) \vartriangleright t$
\item Pearl counterfactuals: 
$\gamma \supset (\SET X = \SET x \cf \Pr(\alpha) \vartriangleright t)$,
\end{enumerate}
which, by the results in sections \ref{subs: conditional probabilities of CO}-\ref{sec: mixed statements}, \corrf{translate} 
 Pearl's notations
\begin{enumerate}[1*.]
\item $\Pr(\alpha\mid\gamma) \vartriangleright t$
\item $\Pr(\alpha\mid do(\SET X = \SET x)) \vartriangleright t$
\item $\Pr(\alpha_{\SET X = \SET x}\mid\gamma) \vartriangleright t$.
\end{enumerate}

\begin{proof}
In order to obtain A., we first push $\cf$ inwards by using the following equivalences:
\begin{align*}
(\SET X = \SET x \cf (\psi\land\chi)) & \equiv  ((\SET X = \SET x \cf \psi) \land (\SET X = \SET x \cf \chi)) \\
(\SET X = \SET x \cf (\psi\sqcup\chi)) & \equiv  ((\SET X = \SET x \cf \psi) \sqcup (\SET X = \SET x \cf \chi)) \\
(\SET X = \SET x \cf (\alpha\supset\chi)) & \equiv  ((\SET X = \SET x \cf \alpha) \supset (\SET X = \SET x \cf \chi)) \\
(\SET X = \SET x \land Y=y) \cf \varphi & \equiv  (\SET X = \SET x \land Y=y) \cf Y=y  \tag{when $\SET X = \SET x \land Y=y$ is inconsistent} \\
\SET X = \SET x \cf (\SET Y = \SET y \cf \chi) & \equiv  (\SET X' = \SET x' \land \SET Y = \SET y) \cf \chi 
\end{align*}
where $\SET X' = \SET X \setminus \SET Y$ and $\SET x' = \SET x \setminus \SET y$; and provided $\SET X = \SET x$ and 
 $\SET Y = \SET y$ are consistent.\footnote{A detailed proof of this latter equivalence can be found in the Arxiv draft \cite{BarSan2017}. The proof just shows that the causal (multi)team produced by applying sequentially the interventions $do(\SET X = \SET x)$ and $do(\SET Y = \SET y)$ is the same that is obtained by the single intervention $do(\SET X' = \SET x' \land \SET Y = \SET y)$; the proof is thus essentially independent of the language in which $\chi$ is formulated.}
Secondly, we replace all consequents of the form $X=x$ (resp. $X\neq x$) with $\Pr(X=x)\geq 1$ (resp. $\Pr(X\neq x)\geq 1$).

After achieving A., we apply the following equivalences
\begin{align*}
\alpha \supset (\psi\land\chi) & \equiv (\alpha \supset \psi)\land (\alpha \supset \chi) \\
\alpha \supset (\psi\sqcup\chi) & \equiv (\alpha \supset \psi)\sqcup (\alpha \supset \chi) \\
\alpha \supset (\beta\supset\chi) & \equiv (\alpha \land \beta)\supset \chi \\
\end{align*}
 in order to push $\supset$ inwards, until it is in front of a counterfactual or an atom (and if the latter is a non-probabilistic  atom $\eta$, again replace it with $\Pr(\eta)\geq 1$).
\end{proof}

\noindent We find it to be somewhat surprising that the conditional $do$ expressions (such as $Pr(y \mid do(x),z) = \epsilon$), which have been intensely studied and are at the center of Pearl's ``\emph{do} calculus'', are in a sense, by this normal form result, omissible. The reason is clarified when translating this kind of expression into $\PCO$ as $X=x \cf (Z=z \supset \Pr(Y=y)=\epsilon)$. While the operators $\cf$ and $\supset$ clearly do not commute, the operator $\cf$ \emph{distributes} over $\supset$ 
 producing an expression $(X=x \cf Z=z) \supset (X=x \cf \Pr(Y=y)=\epsilon)$, which is a counterfactual in Pearl's sense -- although the conditioning is made on a condition that describes the outcome of an intervention.

\section{A semantically complete language for probabilistic conditionals}\label{sec: infinitary language}

As we mentioned in section 
 \ref{subs: conditional probabilities of CO}, for the purpose of applications to causal inference one may need more resources than those available to our language $\PCO$. One may then consider extensions of this language in which the needed additional notions, such as probabilistic independence atoms, general conditional comparison atoms or arithmetical operators, are available as primitives. It is then natural to ask whether this process can stop at some point. We show in this section that there is indeed an extension of $\PCO$ in which all concepts relevant to reasoning with probabilistic causal counterfactuals are expressible.\footnote{But not necessarily schematically definable. For a discussion of the split between expressibility and definability in team logics, see for example \cite{CiaBar2019}.} The price for this kind of expressive completeness will be the introduction of an infinitary connective (which immediately makes the language itself uncountable). We can prove, indeed, that no countable language suffices for the task. Fixing a finite signature, and ignoring the complications given by the causal apparatus, any causal multiteam encodes a rational-valued probability distribution over a fixed finite set (the largest team $\B_\sigma$ allowed by the signature). There are countably many such distributions. A formula characterizes a set of causal multiteams, and thus, if we assume that our language is fully general, its formulas must be capable of describing any subset of the countable set of probability distributions over $\B_\sigma$; the language must then have at least the size of the continuum.

We now allow for possibly infinite global disjunctions:
\[ 
\bigsqcup_{i\in I} \psi_i
\]
where it suffices to assume that $I$ is an at most countable index set, and that the $\psi_i$ are $\PCO$ formulas.\footnote{We thank Jonni Virtema for this observation.}\footnote{We remark that the language itself is \emph{not} countable.} We call $\PCO^\omega$ the language allowing for such disjunctions. What we obtain, more precisely, is that such a language can express all statements (about causal counterfactuals) that are purely probabilistic, and agnostic about the possible interpretation of the concept of probability. In other words, such a language will not allow us to tell apart two data populations that encode identical probability distributions (and, thus, to make claims about a frequentist origin of the probabilities). This idea is made precise by describing an appropriate closure condition for families of causal multiteams.
Given an assignment $t\in\B_\sigma$ and a team $T = (T^-,\F)$, we write $\#(t,T)$ for the number of copies of $t$ in $T^-$.
 We can then talk of the probability of $t$ in $T$, 
 \[
 \epsilon_t^T \dfn \frac{\#(t,T)}{|T^-|}.
 \]
  We say that two causal teams $S =   (S^-,\F)$, $T=(T^-,\G)$ are \textbf{rescalings} of each other ($S\sim T$) if $\F=\G$ and either $S^- = T^- = \emptyset$ or $\epsilon_t^T =\epsilon_t^S$ for each assignment $t$. 
\begin{itemize} 
\item A class $\K$ of causal multiteams of signature $\sigma$ is \textbf{closed under rescaling} if, whenever $S\in \K$,  $S \sim T$, then $T\in \K$.
\end{itemize}

\noindent \corrf{Given a set of formulas $\Gamma$ of signature $\sigma$, we write $\K_\Gamma^\sigma$ (or, more simply, $\K_\Gamma$) for the set of all causal multiteams of signature $\sigma$ that satisfy all the formulas in $\Gamma$. We say that a set $\K$ of causal multiteams of signature $\sigma$ is \textbf{definable} in a language $\La$ if there is $\Gamma \subseteq \La$ such that $\K=\K_\Gamma$.} We can now formulate our semantic completeness result for $\PCO$.

\begin{teo}\label{theorem: expressivity of PCOinf}
A \emph{nonempty} class $\K$ of multiteams of signature $\sigma$ is definable in 
 $\PCO^{\omega}_{\sigma}$ 
  iff $\K$ \corrf{contains all the empty causal multiteams of signature $\sigma$} and is closed under rescaling.
\end{teo}

This closure result also tells us that the language $\PCO^\omega$ could in principle be enriched with other useful operators, such as infinitary conjunction, a liberal use of the tensor disjunction or various kinds of (in)dependence atoms, without changing its expressive power; the only limitation is that the additional operators must preserve closure under rescaling and the empty multiteam property. 
We postpone the somewhat long proof of the theorem to the Appendix. 
 The proof extends some of the methods that were developed in \cite{BarYan2022} for the analysis on non-probabilistic causal-observationl languages. The key idea is that any class $\K$ of causal multiteams of the kind mentioned in the statement of the theorem is described by a specific formula in a normal form:
\[
\bigsqcup_{(T^-,\F)\in \K }(\Theta_{T^-} \land \Phi^\F)
\]
where $\Theta_{T^-}$, resp. $\Phi^\F$, are $\PCO$ formulas that express that the multiteam component is $T^-$ (up to rescaling), resp. that the function component is $\F$. Such a formula explicitly describes the structure of a family of causal multiteams. An infinitary disjunction is needed because, even working with a finite signature as we do in this paper, a $\PCO$ formula may happen to be satisfied by a family of causal multiteams that encode an infinite family of distinct probability distributions. (For example, a validity is satisfied -- relative to a certain signature -- by causal multiteams that encode all possible rational-valued probability distributions over a given finite set.)

We remark that the $\PCO^\omega$ formulas can also be written in a Pearl-style normal form, as in theorem \ref{thm: Pearl normal form}, allowing of course for infinitary Boolean combinations. It turns out, then, that conditional probability statements, unconditional $do$ expressions and Pearl counterfactuals can be taken as the building blocks of causal reasoning also in this most general language; this classification is not an artifact of our focus on the language $\PCO$.

We end this section by pointing out that the conditions of closure under rescalings and the empty team property are not trivial constraints given by the semantics, but they depend on a careful choice of the logical operators and atoms. There are some seemingly reasonable choices that immediately spoil these properties. In many publications (e.g. \cite{RasOgnMar2004}) a more liberal definition of probabilistic atoms is given, by which any factual formula $\alpha$ is stipulated to have probability $1$ in empty models. This is intuitively justified by the trivial reading of quantifiers, since we have that \emph{all} assignments in an empty model satisfy $\alpha$.  We can express this interpretation of probabilistic atoms as follows:
\begin{itemize}
\item $T\models \Pr^*(\alpha)\geq \epsilon$ iff $P_T(\alpha)\geq \epsilon$ or $T$ is empty 
\item $T\models \Pr^*(\alpha)\leq \epsilon$ iff $P_T(\alpha)\leq \epsilon$ or $T$ is empty and $\epsilon=1$
\item $T\models \Pr^*(\alpha)> \epsilon$ iff $P_T(\alpha) >\epsilon$ or $T$ is empty and $\epsilon<1$
\item $T\models \Pr^*(\alpha)< \epsilon$ iff $P_T(\alpha) < \epsilon$. 
\end{itemize}
An immediate consequence of these definitions is that a language allowing such probabilistic atoms can express a \emph{nonemptyness atom} NE, as follows:
\begin{itemize}
\item $T\models \mathrm{NE}$ if and only if $T\models \Pr^*(X=x) < 1 \sqcup \Pr^*(X\neq x) < 1$.
\end{itemize}
Indeed, if a causal multiteam is empty, then it does not satisfy probabilistic atoms that use the $<$ operator; vice versa, if it is not empty, then at least one of the formulas $X=x$, $X\neq x$ is false on some assignment. Now, then, the formula NE violates the empty multiteam property.  

We can obtain a violation of closure under rescalings if we also assume that the language admits free usage of the \emph{strict tensor disjunction}, i.e. the operator $\curlyvee$ given by:
\begin{itemize}
\item $T\models \psi \curlyvee \chi$ iff  if there are two causal sub-multiteams $T_1,T_2$ of $T$ such that $T_1^-\cap T_2^- = \emptyset$, $T_1^-\cup T_2^- = T^-$, $T_1\models \psi$ and $T_2\models \chi$.
\end{itemize}
We remark that $\curlyvee$ is not a particularly exotic kind of connective; it is the form of disjunction that is most commonly used in the the literature on multiteam semantics.  
We can now find a formula $\theta_k$ that expresses the property of a causal multiteam (of a signature $\sigma$) of having at least $k$ assignments. 
 Such formula is simply:
\[
\theta_k: \bigcurlyvee_{k \ times} \mathrm{NE}
\]
 Now take a causal multiteam $S$ (of signature $\sigma$) of cardinality $k-1$, and let $T$ be a rescaling of $S$ of cardinality $> k$ (as larger rescalings are always easy to produce, see the Appendix). Then $S\not\models \theta_k$, while $T\models \theta_k$, i.e. the family of causal multiteams  defined by $\theta_k$ is not closed under rescaling.
 


\section{The debate on the ladder of causation}\label{sec: debate}

Pearl's classification of the tasks of causal inference according to the ``ladder of causation'' has not gone without critics. We want to examine one such point, raised by Tim Maudlin in \cite{Mau2019}.
 Maudlin's critique to Pearl's ``ladder of causation'' is that the distinction between level two (interventions) and level three (counterfactual reasoning) in the hierarchy does not seem to be justified. Counterfactual reasoning is already involved at the level of interventions in the very definition of what an intervention is; and thus the distinction between the two rungs of the ladder is blurred. Pearl answered publicly\footnote{In a page of his blog ``Causal Analysis in Theory and Practice'', http:// causality.cs.ucla.edu/blog/index.php/2019/09/09/on-tim-maudlins-review-of-the-book-of-why/.} to this criticism by pointing out that (although this may not be evident from his book) levels two and three differ strongly in one technical sense. While typically a counterfactual statement can be evaluated only if we have full knowledge of the causal laws (structural equations) which affect its truth value (or degree of probability), it turned out, somewhat surprisingly, that the probabilities of the simpler counterfactuals involved at level two of the ladder of causation can always be computed from observational data by mere knowledge of the causal graph. In other words, one does not need to know the specific functional form of the causal laws, but only which variables affect which; this is essentially the main take of the renowned completeness theorem for Pearl's \emph{do} calculus, \cite{ShpPea2006,HuaVal2006}. We find this fact to be of great theoretical interest, but not a solid enough defense of the structure of the ``ladder''. To the best our knowledge, it has not been proved that this result is maximal, i.e. that it might not be extended, in the future, to a larger class of counterfactuals than the class of those involved at rung two of the ladder.\footnote{An extension of the range of purely graphical methods has already been accomplished in \cite{ShpPea2007}, but it concerns estimation of the probabilities of counterfactuals from \emph{experiments} -- while the $do$ calculus concerns estimation from \emph{observational} data. As far as we know, the estimation of counterfactuals from observational data has not yet been systematically investigated.} Thus, the boundary between the counterfactuals that can be estimated from observational data and causal graph only, and those that cannot, is still undetermined and cannot at the present state of knowledge serve as a justification for the distinction between rung two and three. A negative answer to the issue of maximality might instead suggest a shift of the boundary between rung two and three, i.e. classifying more counterfactuals as pertaining to rung two. The developments of the present paper support, in a different way, the idea that the boundary between rung two and three of the ladder of causation is questionable. Our arguments come from the perspective of logic. We have seen that 
  statements concerning the most typical  probabilistic quantities used in causal inference can be ultimately decomposed in terms of (ordinary connectives plus) two conditional operators $\supset$ (selective implication) and $\cf$ (interventionist counterfactual), the former describing the consequences of an increase of information (for example through observation) and the latter the consequences of an action. It turns out that the statements in the first rung of the ladder (the level of associations) are clearly distinguished from those in the upper levels, in that, after having been translated into the $\PCO$ language, they involve only the conditional $\supset$. Both kinds of conditionals are involved at level two as well as at level three:
  at level two we have the conditional $do$ expressions that translate in $\PCO$ to formulas of the form $\SET X=\SET x \cf (\alpha\supset\psi)$, while at level three we have the Pearl counterfactuals that translate into formulas of the form $ \alpha\supset(\SET X=\SET x \cf\psi)$. Thus, from a logical point of view, no qualitative difference is discernible between  layers two and three, except for the order in which the two conditionals are applied. 
  
  A different perspective is suggested by the normal form theorem (\ref{thm: Pearl normal form}) that we have proved for $\PCO$. This result, on one hand, seems to confirm that the scope of Pearl's classification is appropriate: while (as exemplified in section \ref{sec: mixed statements}) a much wider class of probabilistic causal statements is in principle definable in $\PCO$, these are ultimately reducible to Boolean combinations of (the $\PCO$ translations of) the expressions considered by Pearl. More precisely, though, the theorem says that such formulas are Boolean combinations of (translations of) conditional probabilistic statements, \emph{un}conditional \emph{do} expressions and Pearl counterfactuals. This result is thus presenting a different hierarchy, in which the second level is narrower than what suggested by Pearl.\footnote{The fact the \emph{conditional} \emph{do} expression are unnecessary from the point of view of expressivity is surprising, given e.g. their vital role in the \emph{do} calculus.} The distinction between the three levels of our hierarchy is then corroborated by the fact that level one involves only the conditional $\supset$, the second level involves only $\cf$, and the third level both $\supset$ and $\cf$. This classification does not seem to be an artifact of our choice of language, since it persists when we consider an extended (infinitary) language $\PCO^\omega$ which provably can express  probabilistic causal statements with maximum generality (see the discussion 
   in section \ref{sec: infinitary language}). 

\section{Conclusions}

We have seen that the notion of causal multiteam that we have introduced here is a way of modeling causal system which is, under many respects, equivalent to the use of semi-deterministic SEMs. We have briefly observed that this equivalence holds when the causal graph under discussion is acyclic; the two kinds of model can diverge otherwise. A precise treatment of the cyclic case is postponed to future research.

Some interesting facts emerged from our analysis of a causal-probabilistic language supported by causal multiteam semantics; these observations probably underly the usual analysis of SEMs, but, as far as we know, they are not usually explicitated in such studies. First of all, we noticed the different roles played by two distinct kind of disjunction: the tensor disjunction $\lor$, which allows the formation of disjunctive events; and the global disjunction $\sqcup$, which is used in disjunctive statements \emph{about} probabilities. \corrf{A similar distinction emerged for other connectives (selective implication $\supset$ versus material implication $\rightarrow$, dual negation $\neg$ versus weak contradictory negation $^C$).} Secondly, we noticed that conditional probabilities can be reduced to marginal probabilities by means of a logical operator $\supset$ that we call selective implication; and, more generally, mixed causal-conditional expressions can be decomposed and better analyzed in terms of $\supset$ and the counterfactual $\cf$. This kind of analysis also reveals that more complex kinds of conditioning are possible than those usually considered in the literature on causation; the simplest case being that of probabilities of a second experiment conditioned upon the outcome of a first experiment. On the other hand, we have shown that, by simple logical manipulations, even the more complex of such expressions can be reduced to Boolean combinations of conditional probability statements, $\emph{do}$ expressions and Pearl counterfactuals. We noticed that, surprisingly, the \emph{conditional} \emph{do} expressions can be eliminated. 

From the point of view of definability theory, we have identified a fully general, infinitary language $\PCO^\omega$ for probabilistic causal reasoning, and pointed out that the above classification persists also in this context. We remarked instead on some limitations of the basic language $\PCO$, such as the inexpressibility of certain conditional comparison atoms. The proof of this result, and a precise characterization of the expressive power of $\PCO$, require methods that go beyond the scope of the present paper, and will be addressed in the forthcoming \cite{BarVir2023}. These methods seem not solve the issue whether probabilistic conditional independence is expressible in $\PCO$. Another issue that is not settled here is the proof theory of languages like $\PCO$ and $\PCO^\omega$. 
 \corrf{For the moment we only know that $\PCO$ admits an infinitary axiomatization (\cite{BarVir2023b}).} 
 

An issue that deserves further study is the degree of applicability of  causal multiteam semantics in realistic situations, when the multiteam is obtained as a set of empirical observations. Consider the following table:


\begin{center}
\begin{tabular}{|c|c|c|}
\hline
\multicolumn{3}{|c|}{ \hspace{-4pt}Key \hspace{2pt} $X$ \hspace{8pt} $Y$} \\
\hline
 \phantom{a}0\phantom{a} & 0 & 1.02 \\
\hline
 1 & 1 & 2.01 \\
\hline
 2 & 1 & 2.02 \\
\hline
 3 & 2 & 3.01 \\ 
\hline
 4 & 2 & 3.02 \\ 
\hline
 5 & 2 & 3.02 \\ 
\hline
\end{tabular}
\end{center}

\noindent This multiteam does not satisfy the dependency $\dep{X}{Y}$, since, for example, the records $1$ and $2$ have the same value for $X$, but different values for $Y$; therefore, there is no function which generates $Y$ in terms of $X$. It is impossible, thus, to add to this multiteam the assumption that ``$X$ is a cause of $Y$'': the pair $(X,Y)$ cannot be an edge of a causal graph associated with this multiteam. However, this prohibition might be too restrictive. If the table comes from empirical observations, it might be the case that the difference in the values of $Y$ for records $1$ and $2$ is simply due to measuring error, while in the real world there is indeed a law that determines $Y$ as a function of $X$. 
It would then be reasonable to considere a less restrictive notion of causal multiteam. Assuming, as is typically the case for empirical data, that the set of values for the variables in the system has a metric (call it $dist$), for any error threshold $\Delta\in\mathbb{R}^+$ we may define a ``$\Delta$-tolerant causal multiteam'' to be a pair $(T^-,\F)$ satisfying the conditions 1. and 2. of the definition of causal multiteam, plus the constraint:
\begin{itemize}
\item For all $s\in T^-$, and all $Y\in \SET V$, \hspace{3pt}  $dist(s(Y),f_Y(s(PA_Y))) \leq \Delta$.
\end{itemize}
In the presence of a metric, then, the causal multiteams coincide with the special case of 0-tolerant causal multiteams. The sharp notions of (causal or contingent) dependence can also be approximated by error-tolerant versions; in the simpler context of team semantics, this kind of generalization has been investigated in \cite{GraHoe2018}.
We leave it for future work to see to what extent the considerations put forward in the present paper can be extended to error-tolerant causal multiteams.


\bibliographystyle{acm}
\bibliography{iilogics}

\appendix

\begin{center}

\vspace{30pt}

\LARGE \textbf{APPENDIX}
\end{center}

\section{Proof of the characterization theorem for $\PCO^\omega$} \label{Appendix: infinitary completeness}

\normalsize

We present here a proof of theorem \ref{theorem: expressivity of PCOinf} from the main text, i.e. we show that $\PCO^\omega$ is characterized by closure under rescaling (plus the empty multiteam property).

We begin by remarking that the formulas without occurrences of probabilistic atoms are in a sense determined just by the support of the causal multiteam, i.e. the set of assignments that are assigned nonzero probability by the multiteam component. Let us make things precise.
We write $\models$ for the satisfaction relation over causal multiteams, and $\models^{ct}$ for the satisfaction relation over causal teams (when there is a need to make the distinction).
The satisfaction clauses of the causal team semantics for language $\CO$ (and for the operator $\sqcup$), as given in previous literature, are formally identical to those of causal multiteam semantics. 





Given a multiteam $T^-$ of signature $\sigma=(Dom,Ran)$, there is a corresponding team $Team(T^-)\dfn \{s_{\upharpoonright Dom} \mid s\in T^-\}$ of signature $\sigma$ (its \textbf{support}). More generally, given a causal multiteam  $T = (T^-,\F)$ of signature $\sigma$, there is a corresponding causal team $Team(T) = (Team(T^-), \F)$.

It is then tedious but easy to prove the following lemma, which allows to translate results between causal team semantics and causal \emph{multi}team semantics. 

\begin{lem}\label{lemma: transfer}
Let $T$ be a causal multiteam of signature $\sigma$, and $\varphi$ a formula of $\CO$, or a formula of $\PCO$ without occurrences of probabilistic atoms. 
Then:
\[
T\models\varphi \iff Team(T)\models^{ct} \varphi
\]
\end{lem}

\noindent It's also easy to prove that $\PCO^\omega$ has the empty team property.

\begin{lem}\label{lemma: empty multiteam property}
Let $S$ be an empty causal multiteam of signature $\sigma$ and $\varphi\in\PCO^\omega_{\sigma}$. 
Then $T\models \varphi$. 
\end{lem}


\noindent Let us now consider the condition of closure under rescaling. It is worth to spell out the following alternative characterization of rescaling, whose proof is straightforward from the definition.

\begin{lem}\label{lemma: characterization of rescaling}
Let $S=(S^-,\F)$, $T=(T^-,\F)$ be nonempty causal multiteams of the same signature and $q= \frac{|T^-|}{|S^-|}$. Then $S \sim T \iff \#(s,T) = q \cdot \#(s,S)$ for each $s\in \B_\sigma$.
\end{lem}

\noindent It is obvious that, for any nonempty causal multiteam $S$ and any $n\in \mathbb{N^+}$, there is a unique causal multiteam $T$ of the same signature such that $S \sim T$ and $\frac{|T^-|}{|S^-|} = n$. We also write $T$ as $nS$, and we say that 
 $nS$ is a \textbf{multiple} of $S$. If $T = mR$ and $T=nS$, where $m,n\in \mathbb{N}^+$, we say $T$ is a \textbf{common multiple} of $R$ and $S$.

\begin{lem}[Existence of common multiples]\label{lemma: common multiple}
Let $R,S$ be nonempty causal multiteams of the same signature. If $R \sim S$, then they have a common multiple.
\end{lem}

\begin{proof}
Let $m$ be the least common multiple of $|R^-|$ and $|S^-|$. Thus $\frac{m}{|R^-|}\in \mathbb{N}^+$. Then $T\dfn\frac{m}{|R^-|} R$ is a multiple of $R^-$. We will show that $T= \frac{m}{|S^-|} S$ and is thus a common multiple of $R$ and $S$, since $\frac{m}{|S^-|}\in \mathbb N$. This amounts to showing that, for each $s \in \B_\sigma$, $\#(s,T) =
\frac{m}{|S^-|} \cdot\#(s,S)$.
Now: 
\begin{align*}
\#(s,T) & = \frac{m}{|R^-|}\#(s,R)  \tag{by $T=\frac{m}{|R^-|}R$} \\
 & =  \frac{m}{|R^-|}\cdot\frac{|R^-|}{|S^-|}\cdot \#(s,S)  \tag{by $R^-\sim S^-$ and lemma \ref{lemma: characterization of rescaling}}\\
 & = \frac{m}{|S^-|} \cdot\#(s,S). \\ 
\end{align*}

\vspace{-20pt}

\end{proof}

\noindent We need two further lemmas showing that rescaling is preserved by interventions and observations.

\begin{lem}\label{lemma: rescaling preserved by interventions}
Let $S=(S^-,\F), T=(T^-,\F)$ be causal multiteams of the same signature, and $S\sim T$. Then $S_{\SET X = \SET x} \sim T_{\SET X = \SET x}$.
\end{lem}

\begin{proof}
By Lemma \ref{lemma: common multiple}, since $S\sim T$, they have a common multiple. Thus, it suffices to prove the simpler statement: if $T = nS$, then $T_{\SET X = \SET x} = nS_{\SET X = \SET x}$.

Let $s\in (T_{\SET X = \SET x})^-$. Write $I_s$ for the set $\{t\in Team(T^-) \mid t^\F_{\SET X = \SET x}=s\}$. Then
\[
\#(s,T_{\SET X = \SET x})  = \sum_{t\in I_s} \#(t,T)
  = \sum_{t\in I_s} n \cdot \#(t,S) 
  = n\cdot \sum_{t\in I_s} \#(t,S) 
  = n\cdot\#(s,S_{\SET X = \SET x}). 
\]
 The first and fourth equalities are justified by the fact that interventions preserve the sizes of multiteams.
\end{proof}

\begin{lem}\label{lemma: rescaling preserved by selections}
Let $S=(S^-,\F), T=(T^-,\F)$ be causal multiteams of the same signature $\sigma$, $\alpha\in\CO_\sigma$, and $S\sim T$. Then $S^\alpha \sim T^\alpha$.
\end{lem}

\begin{proof}
Let $s\in \B_\sigma$. If $\#(s,S^\alpha)=0$, then either $\#(s,S)=0$ or, for every $t\in S^-$ such that $t_{\upharpoonright Dom}=s$, $(t,\F)\not\models \alpha$. In the former case, since $S\sim T$, $\#(s,T)=0$, and thus $\#(s,T^\alpha)=0$. In the latter case, 
we immediately get $\#(s,T^\alpha)=0$.

If instead $\#(s,S^\alpha)>0$, then $\#(s,S^\alpha)= \#(s,S)$ (since $\CO$ cannot tell apart distinct copies of the same assignment, a consequence of lemma \ref{lemma: transfer}). 
Since $S\sim T$, $\#(s,S)=\#(s,T)$. Finally, again by lemma \ref{lemma: transfer}, 
 this equals $\#(s,T^\alpha)$.
\end{proof}

\begin{lem}[Closure under rescaling]\label{lemma: closure under rescaling}
Let $S=(S^-,\F), T=(T^-,\F)$ be causal multiteams of signature $\sigma$, $S \sim T$, $\varphi\in\PCO^{\omega}_{\sigma}$, and $S\models \varphi$. Then $T\models \varphi$.
\end{lem}

\begin{proof}
We prove the statement (simultaneously for all pairs $S,T$ of causal multiteams) by induction on $\varphi$.

\begin{itemize}






\item Case $\varphi$ is $\Pr(\alpha)\geq \epsilon$. (The case for $>$ is analogous.) If $S\models \Pr(\alpha)\geq \epsilon$, this means 
 that   $\frac{|(S^\alpha)|^-}{|S^-|} \geq \epsilon$, i.e that $\sum_{s\in Team((S^\alpha)^-)} \epsilon_s^S \geq \epsilon$. By $S\sim T$, we have $\epsilon_s^S = \epsilon_s^T$ for each such $s$, and thus $\sum_{s\in Team((S^\alpha)^-)} \epsilon_s^T \geq \epsilon$. But since $Team(S^-)=Team(T^-)$ (by $S^-\sim T^-$), we also have   $Team((S^\alpha)^-) = Team((T^\alpha)^-)$. 
  So, $\sum_{s\in Team((T^\alpha)^-)} \epsilon_s^T \geq \epsilon$, from which $\frac{|(T^\alpha)^-|}{|T^-|}\geq \epsilon$, i.e. $T\models\Pr(\alpha)\geq \epsilon$, follows.

\vspace{5pt} 

\item If $\varphi$ is $X=x$ (resp. $X\neq x$) then it is equivalent to $\Pr(X=x)\geq 1$ (resp. $\Pr(X\neq x)\geq 1$), so this case is reduced to the previous one.

\vspace{5pt} 

\item Case $\varphi$ is $\Pr(\alpha)\geq \Pr(\beta)$. (The case for $>$ is analogous.) As in the first case, one can prove that $S\models\Pr(\alpha)\geq \Pr(\beta)$ iff $\sum_{s\in Team((S^\alpha)^-)} \epsilon_s^S \geq \sum_{s\in Team((S^\beta)^-)} \epsilon_s^S$. But since $S\sim T$, $\epsilon_s^S = \epsilon_s^T$ for all $s\in\B_\sigma$; and, as before, $Team((S^\alpha)^-) = Team((T^\alpha)^-)$ and $Team((S^\beta)^-) = Team((T^\beta)^-)$. Thus, the  condition is equivalent to $\sum_{s\in  Team((T^\alpha)^-)} \epsilon_s^T \geq \sum_{s\in  Team((T^\beta)^-)} \epsilon_s^T$, and finally to $T\models\Pr(\alpha)\geq \Pr(\beta)$.

\vspace{5pt}

\item The cases for $\varphi$ of the forms $\psi\land \chi$, $\psi\sqcup \chi$ or $\bigsqcup_{i\in I}\psi_i$ are straightforward.

\vspace{5pt} 

\item Case $\varphi$ is $\SET X = \SET x \cf \psi$. Let $S\models \SET X = \SET x \cf \psi$. Then $S_{\SET X = \SET x}\models \psi$. Since $S\sim T$, by lemma \ref{lemma: rescaling preserved by interventions} we have $S_{\SET X = \SET x} \sim T_{\SET X = \SET x}$. Then, by the inductive hypothesis, $T_{\SET X = \SET x}\models \psi$. Thus  $T\models \SET X = \SET x \cf \psi$.

\vspace{5pt}

\item Case $\varphi$ is $\alpha\supset \chi$. Suppose $S\models \varphi$. Then $S^\alpha\models \chi$. Since $S\sim T$, by lemma \ref{lemma: rescaling preserved by selections} $S^\alpha\sim T^\alpha$. But then by the inductive hypothesis $T^\alpha\models \chi$, and thus $T\models\alpha\supset \chi$.
\end{itemize}
\end{proof}

As was shown in \cite{BarYan2022}, one can associate to every function component $\F$ (of signature $\sigma$) a $\CO_{\sigma}$ formula $\Phi^\F$ that is satisfied precisely by those causal teams that have a function component similar to $\F$ ($\G$ is similar to $\F$ if they have the same set of endogenous variables, and for all endogenous $V$, $\F_V$ and $\G_V$ differ at most for having a different set of dummy arguments). This result lifts to causal \emph{multi}teams, and it gets simpler, since in our semantics a function component $\F$ is similar only to itself. Write $End(\F)$ for the set of endogenous variables of $\F$. Then, with the notations of this paper\footnote{The original formula mentioned the set of constant causal functions, which are not allowed here. Furthermore, it had to refer explicitly to the parent set of causal functions.} the formula reads:
\[
\Phi^\F: \bigwedge_{V\in End(\F)} \eta_\sigma(V) \land \bigwedge_{V\notin End(\mathcal F)  
} \xi_\sigma(V)
\]
where
\[
\eta_\sigma(V): \bigwedge_{\SET w\in Ran(\SET W_V)}(\SET W_V = \SET w \cf V = \F_V(\SET w))
\]
and
\[
\xi_\sigma(V): \bigwedge_{\substack{\SET w\in Ran(\SET W_V) \\ v \in Ran(V)}} V=v \supset (\SET W_V=\SET w \cf V=v).
\]
Notice that $\Phi^\F$ has no occurrences of $\lor$, and so it is also a $\PCO$ formula.

\begin{lem}[Capturing causal laws]\label{lemma: capturing causal laws}
Let $T = (T^-,\G)$ be a causal multiteam of signature $\sigma$, with $T^-\neq \emptyset$. Then
\(
T\models \Phi^\F  \iff \G = \F.
\)
\end{lem}

\noindent This can be proved by first showing that this characterization holds for causal teams (despite the -- unsubstantial --  technical differences between our definitions and those in previous literature) by a similar proof as that given for theorem 3.4  in \cite{BarYan2022}; and then using the fact that $T\models \Phi^\F$ iff $Team(T)\models^{ct} \Phi^\F$ (lemma \ref{lemma: transfer}).


 Similarly, there is a family of $\PCO$ formulas that characterize multiteams up to rescaling. Some notations that we introduced earlier for causal multiteams also make sense for multiteams \emph{simpliciter}: if $T^-$ is a multiteam of signature $\sigma$ and $s\in\B_\sigma$, $\#(s,T^-)$ will stand for the number of copies of $s$ in $T^-$. We can write $\epsilon_s^{T^-}$ for $\frac{\#(s,T^-)}{|T^-|}$. Finally, we will write $S^-\sim T^-$ if $\epsilon_s^{S^-} = \epsilon_s^{T^-}$ for all $s\in \B_\sigma$ (or $S^-=T^-=\emptyset$).  Now, given a multiteam $T^- \neq \emptyset$ of signature $\sigma$, we define:
\[
\Theta_{T^-}\dfn \bigwedge_{s\in \B_\sigma} \Pr\left(\bigwedge_{V\in Dom} V = s(V)\right)= \epsilon_s^{T^-}.
\]
As a special case, we let $\Theta_\emptyset \dfn \bot$. 

\begin{lem}\label{lemma: capturing probability distribution}
Let $S = (S^-,\F)$ be a nonempty causal multiteam of signature $\sigma$, and let $T^-$ be a multiteam of signature $\sigma$. Then:
\[
S\models \Theta_{T^-} \iff S^- \sim T^-.
\]
Furthermore, \corrf{for any $S$,} $S\models \Theta_\emptyset \iff S^- = \emptyset$. 
\end{lem}

\begin{proof}
If $T^- =\emptyset$, then $S\models \Theta_\emptyset = \bot$ iff $S^- =\emptyset$. 

Suppose $T^- \neq\emptyset$. Then, given the assumption that $S^- \neq \emptyset$,
\begin{align*}
S\models \Theta_{T^-} & \iff \text{ for all } s\in\B_\sigma,  P_S\Big(\bigwedge_{V\in Dom} V = s(V)\Big)= \epsilon_s^{T^-}\\
& \iff \text{ for all } s\in\B_\sigma,  \frac{|\{t\in S^- \mid (\{t\},\F) \models \bigwedge_{V\in Dom} V = s(V)\}|}{|S^-|}= \epsilon_s^{T^-}\\
& \corrf{\iff \text{ for all } s\in\B_\sigma,  \frac{|\{t\in S^- \mid t_{\upharpoonright Dom} = s\}|}{|S^-|}= \epsilon_s^{T^-}}\\
& \iff \text{ for all } s\in\B_\sigma,  \frac{\#(s,S^-)}{|S^-|}= \epsilon_s^{T^-}\\
 & \iff   \text{ for all } s\in\B_\sigma, \quad  \epsilon_s^{S^-}= \epsilon_s^{T^-} \\
 & \iff S^- \sim T^-.
\end{align*}
\end{proof}



\begin{proof}[Proof of Theorem \ref{theorem: expressivity of PCOinf}]
$\Rightarrow$) Let $\Gamma\subseteq\PCO^\omega_{\sigma}$. Now $\K_\Gamma$ is 
closed under rescaling by lemma \ref{lemma: closure under rescaling}, and has the empty multiteam property by lemma \ref{lemma: empty multiteam property}.

$\Leftarrow$) Suppose $\K\neq\emptyset$ \corrf{contains all the empty causal multiteams of signature $\sigma$} 
 and is closed under 
 rescaling. 
Define the formula:
\[
\Psi_\K \dfn \bigsqcup_{T=(T^-,\F)\in\K}(\Theta_{T^-} \land \Phi^\F)
\]
where $\Theta_{T^-}$ is as in lemma \ref{lemma: capturing probability distribution}; and $\Phi^\F$ is as in lemma  \ref{lemma: capturing causal laws}, with the exception that, if $T=(\emptyset,\F)$, we take $\top$ 
 in stead of $\Phi^\F$. We prove that $\K = \K_{\Psi_\K}$.

$\subseteq$) Suppose $T = (T^-,\F)\in \K$. Now $T\models \Theta_{T^-}$ (by lemma \ref{lemma: capturing probability distribution}, since $T^-\sim T^-$). In case $T^- = \emptyset$, simply observe that, by the empty multiteam property,  we have $T\models \Theta_{T^-} \land \top$, and thus $T\models \Psi_\K$, (i.e. $T\in \K_{\Psi_\K}$). 
In case $T^-\neq  \emptyset$, we also have $T\models \Phi^\F$ (by lemma \ref{lemma: capturing causal laws}). 
 Thus again $T\models \Psi_\K$. 

$\supseteq$) Suppose $T\in \K_{\Psi_\K}$. Then there is either an $S=(S^-,\G)\in\K$ with $S^-\neq \emptyset$ such that $T=(T^-,\F)\models \Theta_{S^-} \land \Phi^\G$; or an $S=(S^-,\G)\in\K$ with $S^-=\emptyset$ such that $T\models \Theta_{S^-}$, i.e. $T\models\bot$. In the first case, by Lemma \ref{lemma: capturing probability distribution} we have $S^- \sim T^-$, and by Lemma \ref{lemma: capturing causal laws} we have $\F = \G$. 
 Since $\K$ is closed under rescaling, we conclude $T\in \K$.

In the second case, since $T\models \bot$, we have $T^-=\emptyset$; thus $T\in\K$ since $\K$ has the empty multiteam property.
\end{proof}

\end{document}